\documentclass[10pt, doublecolumn]{IEEEtran}
\usepackage{epsfig,latexsym}
\usepackage{float}
\usepackage{indentfirst}
\usepackage{amsmath}
\usepackage{bm}
\usepackage{amssymb}
\usepackage{times}
\usepackage{enumitem}
\usepackage{bm}
\usepackage{algorithm}
\usepackage[noend]{algpseudocode}
\usepackage{subfigure}
\usepackage{psfrag}
\usepackage{hyperref}
\usepackage{cite}
\usepackage{lastpage}
\usepackage{fancyhdr}
\usepackage{color} 
 \usepackage{amsthm}
\usepackage{bigints}
\usepackage{array}
\usepackage{booktabs}
\newtheorem{Lemma}{Lemma}
\newtheorem{Corollary}{Corollary}
\newtheorem{lemma}[Lemma]{$\mathbf{Lemma}$}

\newtheorem{corollary}[Corollary]{$\mathbf{Corollary}$}

\newcounter{problem}
\newcounter{save@equation}
\newcounter{save@problem}
\makeatletter
\newenvironment{problem}
{\setcounter{problem}{\value{save@problem}}%
  \setcounter{save@equation}{\value{equation}}%
  \let\c@equation\c@problem
  \subequations
}
{\endsubequations
  \setcounter{save@problem}{\value{equation}}%
  \setcounter{equation}{\value{save@equation}}%
}

\begin{document}
\title{  \vspace{-0em}{  On the Impact of Pinching Antennas on Traffic Offloading   }}

\author{ Zhiguo Ding, \IEEEmembership{Fellow, IEEE}, Robert Schober, \IEEEmembership{Fellow, IEEE}, and H. Vincent Poor, \IEEEmembership{Life Fellow, IEEE}   \thanks{ 
  
\vspace{-1em}

Z. Ding is with Nanyang Technological University, Singapore. 
R. Schober is with the Institute for Digital Communications,
Friedrich-Alexander-University Erlangen-Nurnberg (FAU), Germany. H. V. Poor is  with the  Department of Electrical and Computer Engineering, Princeton University,
Princeton, NJ 08544, USA.
 

  }\vspace{-2.5em}}
 \maketitle

\begin{abstract}
Pinching antennas are characterized by their capability to create strong line-of-sight connections and realize multi-antenna systems in a flexible manner.  Existing works have demonstrated the significant potential of pinching antennas for physical layer design. The aim of this paper is to investigate how pinching antennas can be used to reshape the architecture of future networks. In particular, this paper is motivated by the key advantage of pinching antennas, which is to reconfigure the physical boundaries of wireless cells, and focuses on the impact of pinching antennas on traffic offloading. The models for traffic offloading and pinching antenna transmission are presented first. Then, two traffic offloading strategies are developed based on whether an offloading user releases its bandwidth in its original cell. An overall transmit power minimization problem is formulated, where the optimal solutions for the transmit powers and antenna locations are obtained. The presented simulation results demonstrate that the use of pinching antennas can efficiently support traffic offloading, yield low energy consumption, and achieve balanced cell resource utilization.
\end{abstract}\vspace{-0.5em}

\vspace{-1em} 

\section{Introduction}
The ever-growing demand for higher system throughput and enhanced transmission reliability motivates the use of flexible antennas, which are distinguished by their capability to reconfigure the wireless channel conditions \cite{you6g,irs1,irs2,10318061,9264694}. Among the potential candidates for flexible antennas, pinching antennas have received significant attention since they can be flexibly placed close to users to establish strong line-of-sight connections between transceivers \cite{pinching_antenna2,mypa}. The generalization of pinching antennas beyond DOCOMO's original design has recently been proposed in \cite{11434944}, where the adoption of guided-wave media other than dielectric waveguides, such as leaky coaxial cables and metallic pipes,  is introduced to support diverse indoor and outdoor applications using both the sub-6 GHz and millimeter-wave bands. Novel designs to improve the efficiency of signal emission and the degrees of freedom of pinching-antenna-enabled transmission have been proposed in \cite{11202577} and \cite{11348983}. Furthermore, the benefits of pinching antennas for other key sixth-generation (6G) enabling techniques, such as physical layer security \cite{11289518, 11215679,11303890}, integrated sensing and communications (ISAC) \cite{mynpj,11212813,11314615,11414134}, federated learning \cite{11360288,hinapa1,yushen1}, etc., have also been recently investigated. 

Compared to the extensive study of the impact of pinching antennas on physical layer design, research on their impact on the network architecture started only recently. For example, stochastic-geometry studies of multi-cell pinching-antenna systems have been carried out in \cite{11195162} and \cite{11315149}, and an algorithmic framework was  developed in \cite{mymulcell} to utilize pinching antennas for realizing low-complexity and low-transmit-power multi-cell coordination. These existing studies revealed that pinching antennas have great potential to flexibly change the shape of a cell, and hence reconfigure the network architecture, as exemplified in Fig. \ref{fig0}. In particular, consider a two-cell scenario, where the two base stations, denoted by ${\rm BS}_0$ and ${\rm BS}_1$, are located at $(-r_c,0)$ and $(r_c,0)$, respectively, where $r_c$ denotes the radius of the cells. Assume that ${\rm BS}_0$ is equipped with a pinching antenna whose location range is between $(-r_c,0)$ and $(0,0)$, and ${\rm BS}_1$ is equipped with a conventional antenna located at $(r_c,0)$. We focus on the users located within the colored semi-circle area in Fig. \ref{fig0a}.  With conventional antennas, the users in the semi-circle area are always served by ${\rm BS}_2$, where Fig. \ref{fig0b} shows the transmit power (in dBm) required to achieve the target data rate of $1$ bit/s/Hz in the conventional-antenna case. Fig. \ref{fig0c} focuses on the pinching-antenna case, where each user in the colored area can be served by either ${\rm BS}_1$ or ${\rm BS}_2$, whichever yields the lower transmit power.   Fig. \ref{fig0} shows that if equipped with a pinching-antenna system, ${\rm BS}_0$ can effectively extend its service region, i.e., it can penetrate the coverage region of ${\rm BS}_1$, such that ${\rm BS}_1$'s users can be offloaded and served by ${\rm BS}_0$ with less power consumption. 

  \begin{figure}[!] \vspace{-2em}
\begin{center}
\subfigure[Considered network toplogies ]{\label{fig0a}\includegraphics[width=0.235\textwidth]{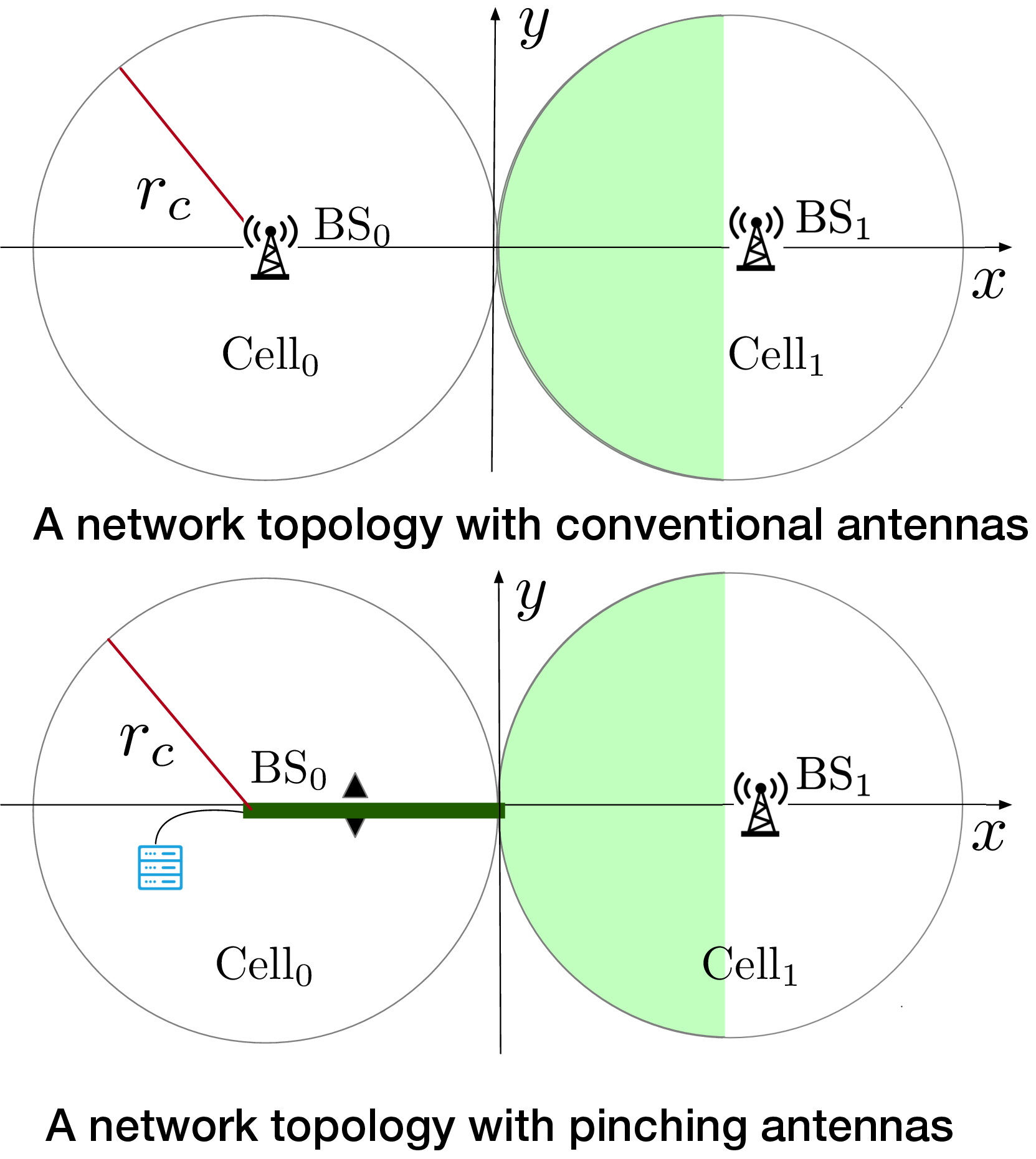}} 
\subfigure[Conventional antennas ]{\label{fig0b}\includegraphics[width=0.23\textwidth]{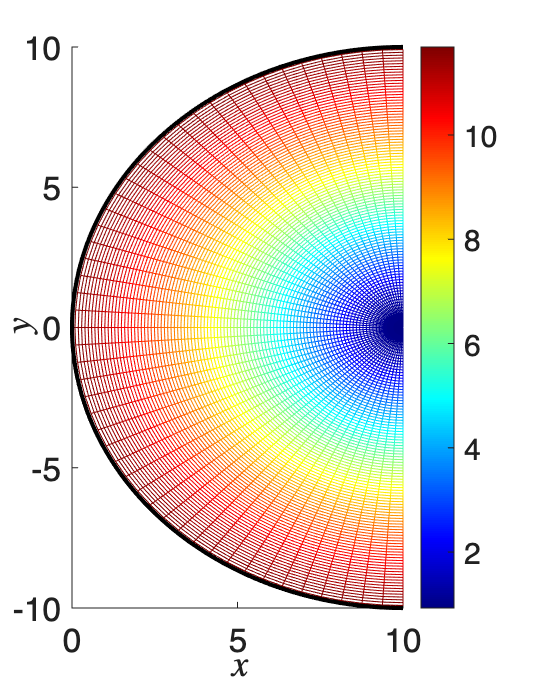}} 
\subfigure[Pinching antennas ]{\label{fig0c}\includegraphics[width=0.23\textwidth]{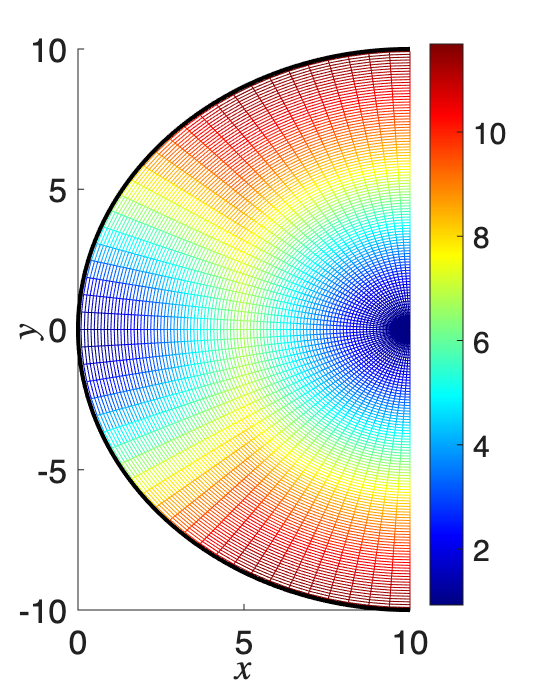}}   \vspace{-1em}
\end{center}
\caption{Illustration of the impact of pinching antennas on cell boundaries. The two base stations, denoted by ${\rm BS}_0$ and ${\rm BS}_1$, are located at $(-r_c,0)$ and $(r_c,0)$, respectively, where $r_c$ denotes the radius of the cells,  $r_c=10$ m, and the target data rate is $1$ bit/Hz/s. ${\rm BS}_0$ is equipped with a pinching antenna whose location range is between $(-r_c,0)$ and $(0,0)$, ${\rm BS}_1$ is equipped with a conventional antenna, and the antennas are placed at a height of $3$ meters. We focus on users located within the given semi-circle. Each user can be served by either ${\rm BS}_0$ or ${\rm BS}_1$, whichever yields the lower transmit power (measured in dBm).   \vspace{-1em} }\label{fig0}\vspace{-1em}
\end{figure}

The significant potential performance gain indicated by Fig. \ref{fig0} motivates us to thoroughly study the impact of pinching antennas on traffic offloading. Recall that traffic offloading is crucial to balance network loads and enhance user experience, particularly in cellular scenarios where some cells are congested while others remain underutilized \cite{7012044}. Supporting traffic offloading becomes even more important for forthcoming 6G applications, which are expected to rely on ultra-dense cell deployments, artificial-intelligence (AI) native control, and extremely high data rates \cite{9061001}. The aim of this paper is to demonstrate that the capability of pinching antennas to reconfigure cell boundaries can be used to efficiently support traffic offloading, achieve low energy consumption, and realize balanced cell resource utilization.

A general traffic offloading model is first developed in the paper, where $M$ overloaded cells try to offload their users to a neighboring underutilized cell with free capacity. Potential applications of the presented traffic offloading model, such as macro-cell-to-small-cell and small-cell-to-small-cell traffic offloading, are illustrated. Then, a system model for pinching-antenna-assisted traffic offloading is established, where the use of a pinching-antenna system with a segmented waveguide is proposed in order to increase the degrees of freedom and reduce maintenance costs.  Whether a user is willing to release its bandwidth in its original cell affects the traffic offloading strategy as follows: 

\begin{itemize}
\item We first consider a simple strategy, where a user does not release its allocated bandwidth in its original cell, i.e., each offloaded user is served in the new cell by using its own bandwidth. The advantage of this strategy is that the $M$ offloaded users can be served by their new base station in an
interference-free manner. However, while this traffic offloading strategy can reduce the overall transmit power of the system, offloading a user from its original cell to the new cell cannot increase the bandwidth available in the original cell. An overall transmit power minimization problem is formulated by optimizing the transmit powers of the base stations and the location of the pinching antenna. Closed-form expressions for the optimal transmit powers are obtained for a given antenna location, and then used to reformulate the considered optimization problem to a convex form, where the optimal antenna location can be obtained efficiently.  Furthermore, a special case with two cells is investigated, where a closed-form expression for the optimal location of the pinching antenna is obtained. The analytical result shows that the optimal location of the pinching antenna is the average of the projections of the users' positions onto the waveguide, where the distances between the users and the waveguide have no impact. 

\item We then consider a more sophisticated traffic offloading strategy, where completely-offloaded users release their allocated bandwidth in their original cell. The key advantage of this strategy is that the departure of a user frees up resources, which allows the user's original cell to accommodate additional users and hence improve the overall resource utilization. However, this offloading strategy is more complex since each offloaded user needs to be served in its new cell by sharing the bandwidth with other users. Non-orthogonal multiple access (NOMA) is employed to facilitate efficient bandwidth sharing while effectively suppressing co-channel interference \cite{mojobabook}. A challenge when applying NOMA is that there are multiple possible successive interference cancellation (SIC) decoding orders. For each of the possible SIC decoding orders, an overall transmit power minimization problem is formulated with the transmit powers and the antenna location as optimization variables. Closed-form expressions for the transmit powers are obtained first and then used to recast the original optimization problem into convex form, which ensures that the optimal solution can be obtained efficiently. For the special case of two cells, concise closed-form expressions for the optimal antenna location and the overall transmit power can be obtained. Unlike the previously considered offloading strategy, the obtained analytical results show that the distances between the users and the waveguide are crucial for determining the optimal SIC decoding order and the optimal location of the pinching antenna.

\end{itemize}

  \begin{figure}[!] \vspace{-2em}
\begin{center}
\subfigure[ A general pinching-antenna assisted traffic offloading scenario ]{\label{fig1a}\includegraphics[width=0.23\textwidth]{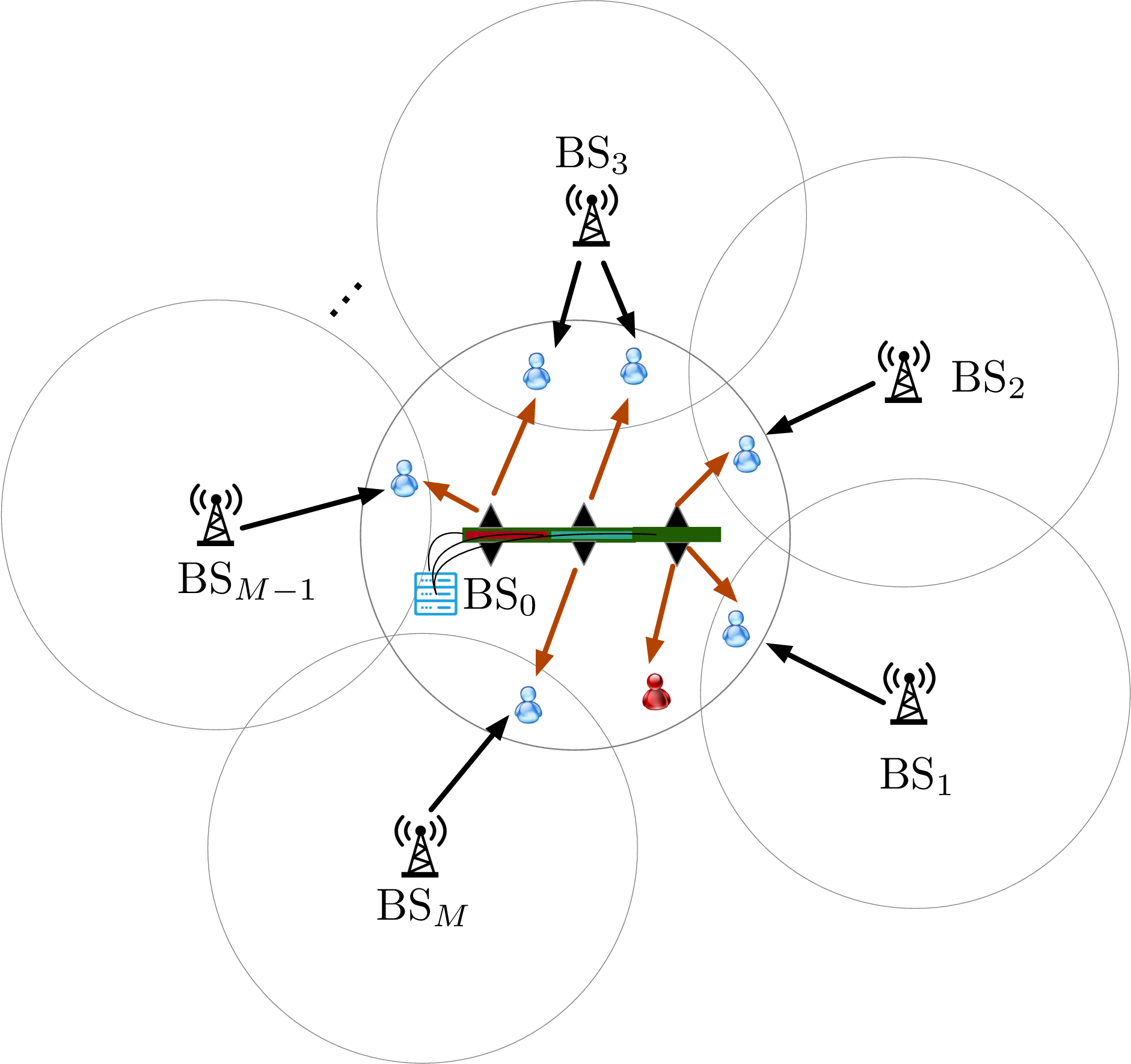}} 
\subfigure[A simplified traffic offloading scenario focusing on a single segement]{\label{fig1b}\includegraphics[width=0.23\textwidth]{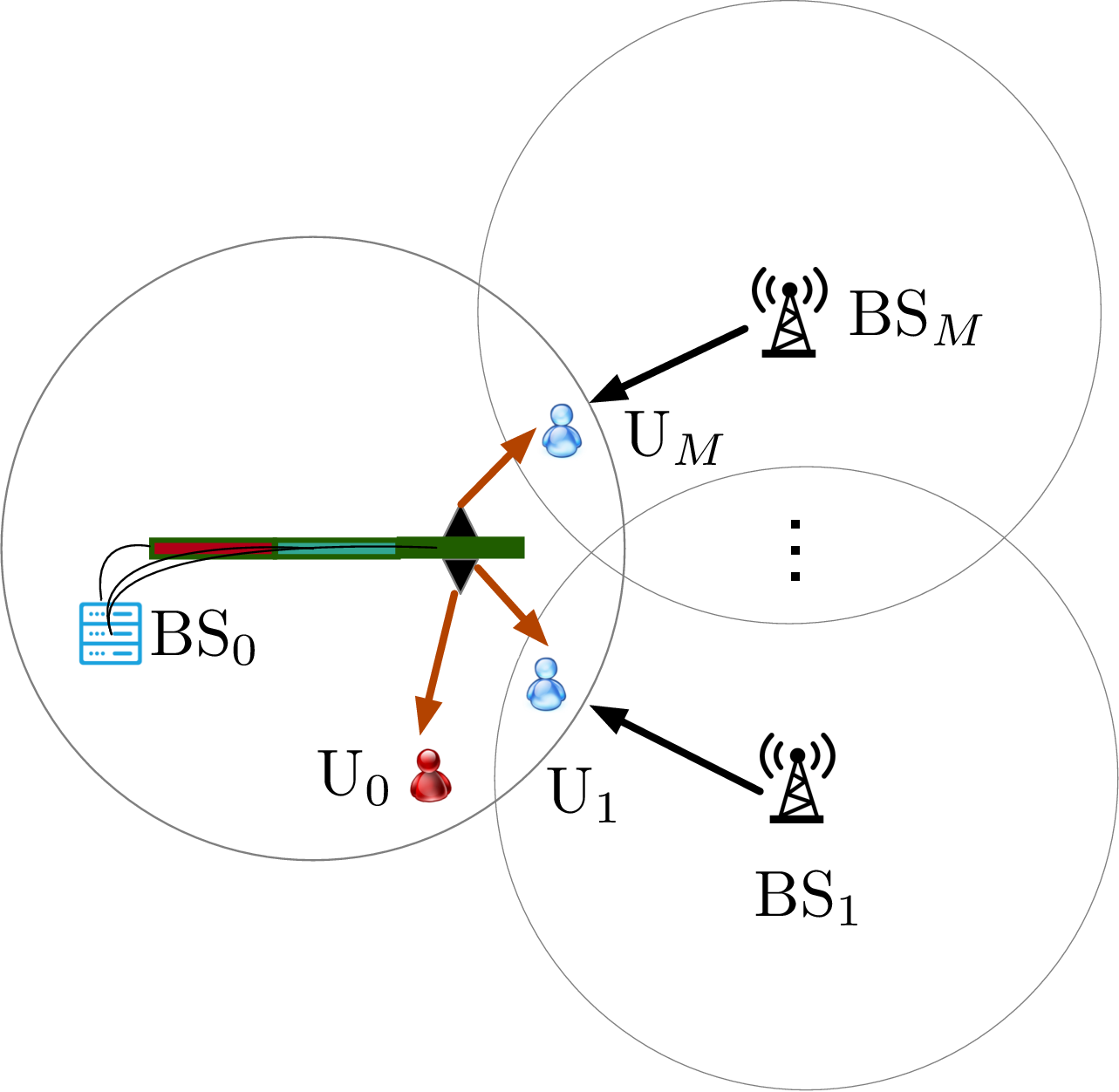}}   \vspace{-1.5em}
\end{center}
\caption{An illustration of the proposed pinching-antenna assisted traffic offloading scheme. \vspace{-1em} }\label{fig1}\vspace{-1em}
\end{figure}


\section{System Model}\label{section models}
\vspace{-0.5em}

\subsection{Traffic Offloading Model}
Consider a multi-cell downlink communication scenario with $M+1$ cells. The cells and their base stations are denoted by ${\rm Cell}_m$ and ${\rm BS}_m$, $0\leq m \leq M$, respectively. We assume that ${\rm Cell}_m$, $1\leq m \leq M$, are overloaded and hence want to offload parts of their traffic to their neighbor, ${\rm BS}_0$, as shown in Fig. \ref{fig1a}.  Applications of the considered traffic offloading scenario are discussed in the following.

\subsubsection{Macro-Cell-to-Small-Cell Traffic Offloading} In heterogeneous networks, each macro base station, e.g., ${\rm BS}_m$, $1\leq m \leq M$, covers a large geographic area, and hence is prone to traffic overloading.  Macro-cell-to-small-sell traffic offloading ensures that heavy traffic in the macro-cell base stations is shifted to a small-cell base station, e.g., ${\rm BS}_0$ \cite{7498101,7883844}. 
 
\subsubsection{Small-Cell-to-Small-Cell Traffic Offloading}
Future 6G systems will feature ultra-dense deployments of small cells, where traffic demand can be unevenly distributed, e.g., some cells are heavily congested while others remain underutilized \cite{9061001}.  A cell with light traffic, e.g., ${\rm Cell}_0$, can help its neighbors by serving the users that are originally associated with heavy-traffic cells, i.e., ${\rm Cell}_m$, $1\leq m \leq M$. As a result, users can be redistributed among cells to ensure load balancing and avoid overcrowding. 

\subsubsection{Cell-to-Wi-Fi Traffic Offloading}
Future wireless systems are expected to exploit the coexistence of cellular systems and Wi-Fi networks \cite{6076617}. Therefore, ${\rm BS}_0$ can be a Wi-Fi access point with light traffic to help reduce the traffic of cellular base stations, ${\rm BS}_m$, $1\leq m \leq M$. In addition, ${\rm Cell}_0$ could also be a cell in an airport/campus shop or building using unlicensed bands, to which traffic of overcrowded licensed bands can be offloaded.


 \vspace{-0.5em}
 \subsection{Pinching-Antenna Assisted Traffic Offloading}
 For illustration purposes, we assume that ${\rm BS}_m$, $1\leq m\leq M$, is equipped with a single conventional fixed-location antenna, and ${\rm BS}_0$ is equipped with a pinching-antenna system. In particular, denote the location of ${\rm BS}_m$, $0\leq m\leq M$, by $ {\boldsymbol \psi}_{m}^{\rm BS}$. 
 
In order to reduce maintenance cost, a segmented waveguide is employed by  ${\rm BS}_0$, where at most a single antenna is activated on each segment, as shown in Fig. \ref{fig1a}. By assuming that cells connected to different segments are served on orthogonal bandwidth resources, e.g., orthogonal frequency-division multiple access (OFDMA) subcarriers, there is no interference among the cells. Furthermore, because each segment is connected to an independent feed point, resource allocation across different segments, such as transmit-power control and antenna placement, can be decoupled \cite{11348983}. Therefore, this paper investigates the optimization of pinching-antenna assisted traffic offloading by focusing on the simplified traffic offloading scenario depicted in Fig. \ref{fig1b}, where a single segment is to offload $M$ users from the $M$ cells adjacent to the segment. Denote the user in ${\rm Cell}_m$ by ${\rm U}_m$, $1\leq m \leq M$. Furthermore, assume that in ${\rm Cell}_0$, there exists a user, denoted by  ${\rm U}_0$, being currently served by the segment. 

The dual connectivity (DC) mode is used, where each ${\rm U}_m$ is served by ${\rm BS}_0$ and ${\rm BS}_m$ simultaneously  \cite{7498101,7883844}. We note that complete offloading, i.e., ${\rm U}_m$ is served by ${\rm BS}_0$ only, is a special case of the DC mode. 
\vspace{-0.5em}
\subsection{Benchmark: No Traffic Offloading}
Without traffic offloading, ${\rm U}_m$ is served by ${\rm BS}_m$ only. We assume that OFDMA is adopted by the cells, where each user is allocated a unique OFDM subcarrier, and there is no interference among the $(M+1)$ users. Therefore, ${\rm U}_m$'s achievable data rate is given by
  \begin{align} 
  R_m^{\rm No-TO}  =
\log_2\left(
1+   \frac{\eta P_m^{\rm No-TO}}{P_{\rm N} \left| {\boldsymbol \psi} _m - {\boldsymbol \psi}_m^{\rm BS}\right|^2}
\right),
\end{align} 
where $ P_m^{\rm No-TO}$ denotes ${\rm U}_m$'s transmit power, $\eta = \frac{c^2}{16\pi^2 f_c^2 }$, $c$ is the speed of light, $f_c$ is the carrier frequency, $P_{\rm N}$ denotes the noise power, ${\boldsymbol \psi} _m$ denotes ${\rm U}_m$'s location, and $\left| {\boldsymbol \psi} _m - {\boldsymbol \psi}_m^{\rm BS}\right|$ denotes the distance between ${\rm U}_m$ and ${\rm BS}_m$. By assuming that all users have the same target data rate, denoted by $R$, the minimal transmit power for ${\rm BS}_m$ to serve ${\rm U}_m$ without traffic offloading is given by $
   P_m^{\rm No-TO} 
 = \epsilon \left| {\boldsymbol \psi} _m - {\boldsymbol \psi}_m^{\rm BS}\right|^2 $,
where $\epsilon = \frac{  P_{\rm N} \left(2^R-1\right)}{\eta}$.

  Depending on how an offloaded user is served in the new cell, two different traffic offloading strategies can be designed, as discussed in the following sections.

 \begin{table}[!]
\centering
\caption{The Two Considered  Traffic Offloading Strategies    \vspace{-1em}}
\begin{tabular}{c|ccc|ccc}
\toprule
  &  \multicolumn{3}{c|}{Strategy I} & \multicolumn{3}{c}{Strategy II}\\
\midrule
&  ${\rm U}_0$  & ${\rm U}_1$ & ${\rm U}_2$&${\rm U}_0$  & ${\rm U}_1$ & ${\rm U}_2$   \\
    \hline
   ${\rm BS}_0$   &$F_0 $&$  F_1 $&$  F_2$ &$F_0 $&$  F_0 $&$  F_0$
    \\
    \hline
     ${\rm BS}_1$    & &$  F_1  $&    & &$  F_1  $&\\
    \hline
 ${\rm BS}_2$   & & &$F_2$ & & &$F_2$ \\ 
\bottomrule
\end{tabular}\label{table1}\vspace{-1.5em}
\end{table}

\section{${\rm BS}_0$ Has Access to ${\rm U}_m$'s Subcarrier in ${\rm Cell}_m$}\label{section III}
This section focuses on the offloading strategy, where ${\rm BS}_0$ has access to ${\rm U}_m$'s subcarrier during traffic offloading, i.e., Strategy I illustrated in Table \ref{table1}. In particular, for the illustrated three-cell special case, $F_m$ denotes the subcarrier allocated to ${\rm U}_m$ in ${\rm Cell}_m$. For Strategy I, during the offloading, ${\rm U}_m$, $1\leq m\leq M$, is simultanesouly served by ${\rm BS}_m$ and ${\rm BS}_0$ on $F_m$. This offloading strategy offers the benefit that the $(M+1)$ users can be served in ${\rm Cell}_0$ in an interference-free manner, but suffers from the disadvantage that the bandwidth available in ${\rm Cell}_m$ is not increased despite the departure of ${\rm U}_m$. 

By using the DC mode, ${\rm U}_m$ is served by ${\rm BS}_m$ and ${\rm BS}_0$ simultaneously on a dedicated subcarrier, with the corresponding transmit powers denoted by $P_m$ and $P_{mm}$, $1\leq m\leq M$, respectively. Therefore, ${\rm U}_m$'s achievable data rate, $1\leq m\leq M$, is given by
  \begin{align} 
  R_m  =&  
\log_2\left(
1+ \frac{  \frac{\eta P_{mm}}{ \left| {\boldsymbol \psi} _m - {\boldsymbol \psi}^{\rm Pin}\right|^2}}
{     \frac{\eta P_m}{ \left| {\boldsymbol \psi} _m- {\boldsymbol \psi}_m^{\rm BS}\right|^2}+P_{\rm N}}
\right) \\\nonumber
&+
\log_2\left(
1+  \frac{\eta P_m}{P_{\rm N} \left| {\boldsymbol \psi} _m - {\boldsymbol \psi}_m^{\rm BS}\right|^2}  
\right)\\\nonumber
 =&  
\log_2\left(
1+   \frac{\eta P_{mm}}{ \left| {\boldsymbol \psi} _m - {\boldsymbol \psi}^{\rm Pin}\right|^2} +  \frac{\eta P_m}{P_{\rm N} \left| {\boldsymbol \psi} _m - {\boldsymbol \psi}^{\rm Pin}\right|^2}  
\right),
\end{align} 
where ${\boldsymbol \psi}^{\rm Pin}$ denotes the location of the pinching antenna, and the last step can be obtained directly from the capacity region of a multiple-access channel \cite{Cover1991}. 

Because the newly arrived users, ${\rm U}_m$, $1\leq m \leq M$, are supported on dedicated subcarriers, the data rate of  ${\rm Cell}_0$'s existing user, ${\rm U}_0$, is the same as for the case without offloading, i.e., $R_0 = 
\log_2\left(
1+   \frac{\eta P_0}{P_{\rm N} \left| {\boldsymbol \psi} _0 - {\boldsymbol \psi}^{\rm Pin}\right|^2}
\right)$, where $P_0$ denotes the transmit power.

 Therefore, the overall transmit power minimization problem can be formulated as follows: 
  \begin{problem}\label{pb:1} 
  \begin{alignat}{2}
 \underset{\mathcal{S}_P,{\boldsymbol \psi}^{\rm Pin}}{\rm{min}}  &\quad   P_0 +\sum^{M}_{m=1}\left(P_m+P_{mm}\right)
\\ s.t. &\quad  \label{1tst:1} R_m\geq R, \quad 0\leq m \leq M ,  \end{alignat}
\end{problem}
where set $\mathcal{S}_P$ collects all users' transmit powers.
For a given location of the pinching antenna, it is straightforward to verify that the optimal solution of $P_0$ is given by $      P_0^*
=\epsilon \left| {\boldsymbol \psi} _0 - {\boldsymbol \psi}^{\rm Pin}\right|^2$, and the optimal solutions of $P_m$ and $P_{mm}$ can be obtained by solving the following $M$ decoupled subproblems: 
   \begin{problem}\label{pb:2} 
  \begin{alignat}{2}
 \underset{P_m,P_{mm},{\boldsymbol \psi}^{\rm Pin}}{\rm{min}}  &\quad   P_m+P_{mm}
\\ s.t. &\quad   \frac{  P_{mm}}{   \left| {\boldsymbol \psi} _m - {\boldsymbol \psi}^{\rm Pin}\right|^2} +  \frac{  P_m}{  \left| {\boldsymbol \psi} _m - {\boldsymbol \psi}_m^{\rm BS}\right|^2} \geq \epsilon,
  \end{alignat}
\end{problem}
 which leads to the following three cases.

\subsubsection{Case with \texorpdfstring
{$\frac{1}{\left|\bm{\psi}_m-\bm{\psi}^{\rm Pin}\right|^2}
>\frac{1}{\left|\bm{\psi}_m-\bm{\psi}_m^{\rm BS}\right|^2}$}
{P12/(|psi2-psi1^Pin|^2) + P2/(|psi2-psi2^Pin|^2)}}
It is straightforward to show that the optimal solution for the transmit power is given by
\begin{align}
P_{mm}^* = \epsilon   \left| {\boldsymbol \psi} _m - {\boldsymbol \psi}^{\rm Pin}\right|^2, 
P_m^*=0.
\end{align}
This case corresponds to the scenario, where ${\rm U}_m$'s traffic is completely offloaded to ${\rm BS}_0$, i.e., ${\rm U}_m$ is served by ${\rm BS}_0$ only.

\subsubsection{Case with \texorpdfstring
{$\frac{1}{\left|\bm{\psi}_m-\bm{\psi}^{\rm Pin}\right|^2}
<\frac{1}{\left|\bm{\psi}_m-\bm{\psi}_m^{\rm BS}\right|^2}$}
{P12/(|psi2-psi1^Pin|^2) + P2/(|psi2-psi2^Pin|^2)}}
This is the case without traffic offloading, where ${\rm U}_m$ is served by ${\rm BS}_m$ only, i.e., 
\begin{align}
P_{mm}^* =0,
P_m^*= \epsilon     \left| {\boldsymbol \psi} _m - {\boldsymbol \psi}_m^{\rm BS}\right|^2,
\end{align}
which means that the overall transmit power for this case is not a function of ${\boldsymbol \psi}^{\rm Pin}$. 

\subsubsection{Case with \texorpdfstring
{$\frac{1}{\left|\bm{\psi}_m-\bm{\psi}^{\rm Pin}\right|^2}
=\frac{1}{\left|\bm{\psi}_m-\bm{\psi}_m^{\rm BS}\right|^2}$}
{P12/(|psi2-psi1^Pin|^2) + P2/(|psi2-psi2^Pin|^2)}} ${\rm U}_m$ is jointly served by both ${\rm BS}_m$ and  ${\rm BS}_0$, i.e., 
\begin{align}
P_{mm}^* +P_m^*= \epsilon     \left| {\boldsymbol \psi} _m - {\boldsymbol \psi}^{\rm Pin}\right|^2.
\end{align}
We note that the expression for the overall transmit power of the third case is identical to those of the other two cases, if $ \left|\bm{\psi}_m-\bm{\psi}^{\rm Pin}\right|^2
= \left|\bm{\psi}_m-\bm{\psi}_m^{\rm BS}\right|^2$. Therefore,  problem \eqref{pb:1} can be equivalently expressed as follows:
  \begin{problem}\label{pb:3} 
  \begin{alignat}{2}
 \underset{{\boldsymbol \psi}^{\rm Pin}}{\rm{min}}  &\quad    \sum_{m\in \mathcal{S}_1} \left| {\boldsymbol \psi} _m - {\boldsymbol \psi}^{\rm Pin}\right|^2
\\ s.t. &\quad  \label{3tst:1}  \left|\bm{\psi}_m-\bm{\psi}^{\rm Pin}\right|^2
\leq  \left|\bm{\psi}_m-\bm{\psi}_m^{\rm BS}\right|^2, \quad m\in\mathcal{S}_1,
\\   &\quad  \label{3tst:2}  \left|\bm{\psi}_i-\bm{\psi}^{\rm Pin}\right|^2
>   \left|\bm{\psi}_i-\bm{\psi}_i^{\rm BS}\right|^2, \quad i\in\mathcal{S}_2, \end{alignat}
\end{problem}
where $\mathcal{S}_1$ denotes the set collecting the users to be completely offloaded to ${\rm Cell}_0$, and $\mathcal{S}_2$ collects the remaining users. 

Assume that the locations of the users and the pinching antenna are given by $\bm{\psi}_m=(x_m,y_m,0)$, and $\bm{\psi}^{\rm Pin}=(x^{\rm Pin},y^{\rm Pin},d)$, respectively. Now, probelm \eqref{pb:3} can be recast as follows:
 \begin{problem}\label{pb:4} 
  \begin{alignat}{2}
 \underset{ x^{\rm Pin} }{\rm{min}}  &\quad     \sum_{m\in \mathcal{S}_1}   \left(
 x^{\rm Pin} -x_m
 \right)^2
\\ s.t. &\quad  \label{3tst:1} x_m-\sqrt{\tau_m}\leq  
 x^{\rm Pin}  \leq  x_m+\sqrt{\tau_m},   \quad m\in\mathcal{S}_1,
\\   &\quad  \label{3tst:2}   x^{\rm Pin} \leq x_i-\sqrt{\tau_i}, {\rm or} ,   
 x^{\rm Pin}  \geq  x_i+\sqrt{\tau_i}, \quad i\in\mathcal{S}_2, \end{alignat}
\end{problem}
where $\tau_{m}=\left|\bm{\psi}_m-\bm{\psi}_m^{\rm BS}\right|^2- \left(
 y^{\rm Pin} -y_m
 \right)^2-d^2
 $.
 
Problem \ref{pb:4} needs to be solved by enumerating all possible choices of $\mathcal{S}_1$. Recall that this paper aims to investigate the impact of pinching antennas on traffic offloading. Hence, the solution with complete offloading is particularly of interest. The optimal antenna location for the complete offloading case can be obtained from problem \eqref{pb:4} by assuming that all users are in $\mathcal{S}_1$, i.e., the constraint in \eqref{3tst:2} is ignored.  We note that without the constraint in \eqref{3tst:2}, problem \eqref{pb:4} is a convex optimization problem, and can be solved efficiently via convex optimization solvers\footnote{Throughout the paper, we assume that complete traffic offloading is feasible. Take problem \eqref{pb:4} as an example. We assume that after the constraint in \eqref{3tst:2} is ignored and all users are in $\mathcal{S}_1$, feasible solutions for the transmit powers and the antenna location exist. In practice, feasibility can be guaranteed by carefully scheduling users, which are close to the waveguide segment, for traffic offloading. }. 

By further assuming that the segment is sufficiently large, constraint \eqref{3tst:1} can be ignored, and problem \eqref{pb:4} can be solved analytically by applying the first-order derivative to its objective function, which leads to the following corollary. 
\begin{corollary}
Assuming that the constraints of problem \eqref{pb:4} can be ignored, the optimal location of the pinching antenna is given by
\begin{align}
\label{analyical 1}
x^{\rm Pin *} = \frac{1}{M}\sum^{M}_{m=0}x_m. 
\end{align}
\end{corollary}

{\it Remark 1:} The optimal analytical solution shown in \eqref{analyical 1} leads to an interesting observation. The optimal antenna location is not a function of $y_m$, the distances between the users and the waveguide, but is a function of the projections of the users' positions onto the waveguide, $x_m$, only. This conclusion is similar to the optimal solutions for the fair resource allocation problems previously investigated in \cite{closedformzid}. 


\section{${\rm BS}_0$ Does Not Have Access to ${\rm U}_m$'s Subcarrier in ${\rm Cell}_m$}\label{section IV}
The offloading strategy considered in the previous section has two shortcomings. First, even if ${\rm U}_m$ is completely offloaded to ${\rm Cell}_0$, the subcarrier originally occupied by ${\rm U}_m$ is not vacant. Therefore, the departure of ${\rm U}_m$ cannot free up bandwidth resources in ${\rm Cell}_m$. Second, the network controller needs to ensure that ${\rm U}_m$'s subcarrier can be used in ${\rm Cell}_0$, which leads to high system complexity. 

Therefore, in this section, we assume that ${\rm BS}_0$ does not have access to ${\rm U}_m$'s subcarrier, i.e., Strategy II illustrated in Table \ref{table1}. For the three-user special case shown in Table \ref{table1}, the subcarrier allocated to an existing user in ${\rm Cell}_0$, i.e., $F_0$, is used to accommodate the offloaded users, which can lead to excessive co-channel interference on $F_0$.  In order to suppress co-channel interference, assume that there are $M$ existing users in ${\rm Cell}_0$, denoted by ${\rm U}_{0m}$. Each of the $M$ new users, ${\rm U}_m$, is paired with an existing user in ${\rm Cell}_0$, e.g., ${\rm U}_{0m}$, and served on ${\rm U}_{0m}$'s subcarrier. As a result, different pairs of users are served on dedicated subcarriers. Compared to the strategy considered in the previous section, complete offloading becomes particularly meaningful, since ${\rm U}_m$'s subcarrier will be released in ${\rm Cell}_m$ and hence the overloading situation in ${\rm Cell}_m$ can be mitigated. Therefore, we are interested in the following two questions:
\begin{itemize}
\item Under what conditions is complete traffic offloading optimal?

\item When complete offloading happens, where should the pinching antenna be placed?
\end{itemize}

Since on each subcarrier two users are served simultaneously, it is advantageous to use NOMA. In particular, on ${\rm U}_{0m}$'s subcarrier, ${\rm BS}_0$ first superimposes the signals for ${\rm U}_{0m}$ and ${\rm U}_{m}$, and then broadcasts the mixture to the users. The adopted SIC decoding order, i.e., how the users carry out SIC, greatly affects the traffic offloading strategy, as discussed in the following two subsections. The optimal transmit powers and the pinching antenna location can be determined by comparing the outcomes of the optimization problems to be formulated for different SIC orders. 
\vspace{-0.5em}

\subsection{${\rm U}_{0m}$ Carries Out SIC} 
If ${\rm U}_{0m}$ carries out SIC, ${\rm U}_{0m}$ can decode ${\rm U}_{m}$'s signal with the following data rate:\footnote{To reduce the system complexity, we assume that the users adopt the same SIC decoding order. Allowing the users to choose different SIC decoding orders can further improve the performance of pinching-antenna-assisted traffic offloading, which is beyond the scope of this paper due to space limitations but is an important direction for future research.    }
\begin{align}
R_{0m}^m=
\log_2\left(
1+ \frac{  \frac{\eta P_{mm}}{ \left| {\boldsymbol \psi} _{0m} - {\boldsymbol \psi}^{\rm Pin}\right|^2}}
{     \frac{\eta P_{0m}}{ \left| {\boldsymbol \psi} _{0m}- {\boldsymbol \psi}^{\rm Pin}\right|^2}+P_{\rm N}}
\right)  ,
\end{align}
where $ {\boldsymbol \psi} _{0m}=(x_{0m}, y_{0m},0)$ denotes ${\rm U}_{0m}$'s location, $P_{mm}$ denotes the transmit power assigned to ${\rm U}_{m}$'s signal by ${\rm BS}_0$ on ${\rm U}_{0m}$'s subcarrier, and $P_{0m}$ denotes the transmit power for ${\rm U}_{0m}$'s signal. 

Provided that the first step of SIC is successful, ${\rm U}_{0m}$ can then decode its own signal with the following data rate: 
\begin{align}
R_{0m}=
\log_2\left(
1+    \frac{\eta P_{0m}}{ P_{\rm N} \left| {\boldsymbol \psi} _{0m}- {\boldsymbol \psi}^{\rm Pin}\right|^2}
\right)  .
\end{align}

Because the DC mode is used, ${\rm U}_{m}$ receives two signals, one from ${\rm BS}_0$ on ${\rm U}_{0m}$'s subcarrier and the other from ${\rm BS}_m$ on its own subcarrier. In particular, on ${\rm U}_{0m}$'s subcarrier, ${\rm U}_{m}$ directly decodes its signal by treating ${\rm U}_{0m}$'s signal as noise, and on its own subcarrier, there is no interference, which means that ${\rm U}_m$'s data rate can be expressed as follows:
\begin{align}
R_m=&
\log_2\left(
1+ \frac{  \frac{\eta P_{mm}}{ \left| {\boldsymbol \psi} _m - {\boldsymbol \psi}^{\rm Pin}\right|^2}}
{     \frac{\eta P_{0m}}{ \left| {\boldsymbol \psi} _m- {\boldsymbol \psi}^{\rm Pin}\right|^2}+P_{\rm N}}
\right) \\\nonumber &+
\log_2\left(
1+  \frac{\eta P_m}{P_{\rm N} \left| {\boldsymbol \psi} _m - {\boldsymbol \psi}_m^{\rm BS}\right|^2}  
\right).
\end{align}

{\it Remark 1:} In order to simplify the following optimization analysis, it is assumed that $R_{0m}^m\geq R_m$, which is equivalent to the following assumption:
\begin{align}\label{sicconstraint1}
 \left| {\boldsymbol \psi} _{0m} - {\boldsymbol \psi}^{\rm Pin}\right|^2\leq \left| {\boldsymbol \psi} _m - {\boldsymbol \psi}^{\rm Pin}\right|^2.
 \end{align} We note that the constraint in \eqref{sicconstraint1} means that the pinching antenna is placed closer to ${\rm U}_{0m}$ than ${\rm U}_m$. This constraint is due to the adopted SIC decoding order, and will be employed by the resource allocation problem formulation.


In particular, the overall transmit power minimization problem is formulated as follows:
   \begin{problem}\label{pb:5} 
  \begin{alignat}{2}
 \underset{ \mathcal{S}_P,  {\boldsymbol \psi}^{\rm Pin}}{\rm{min}}    &\quad   \sum_{m=1}^{M} \left(P_{0m}+P_{mm}+P_{m}\right)
\\ s.t. &\quad   \label{pb:5 con2}
\log_2\left(
1+    \frac{\eta P_{0m}}{ P_{\rm N} \left| {\boldsymbol \psi} _{0m}- {\boldsymbol \psi}^{\rm Pin}\right|^2}
\right) \geq R, \quad 1\leq m \leq M,\\  &\quad \label{pb:5 con3}
\log_2\left(
1+ \frac{  \frac{\eta P_{mm}}{ \left| {\boldsymbol \psi} _m - {\boldsymbol \psi}^{\rm Pin}\right|^2}}
{     \frac{\eta P_{0m}}{ \left| {\boldsymbol \psi} _m- {\boldsymbol \psi}^{\rm Pin}\right|^2}+P_{\rm N}}
\right) \\\nonumber &+
\log_2\left(
1+  \frac{\eta P_m}{P_{\rm N} \left| {\boldsymbol \psi} _m - {\boldsymbol \psi}_m^{\rm BS}\right|^2}  
\right)\geq R,\quad 1\leq m \leq M,\\ &\quad
 \left| {\boldsymbol \psi} _{0m} - {\boldsymbol \psi}^{\rm Pin}\right|^2< \left| {\boldsymbol \psi} _m - {\boldsymbol \psi}^{\rm Pin}\right|^2, \quad 1\leq m \leq M.
  \end{alignat}
\end{problem}
The above optimization problem can be solved by first determining the optimal transmit powers for a given pinching antenna location, and subsequently optimizing the antenna location, as shown in the following two subsections.

\subsubsection{Optimizing the transmit power for a given antenna location}
By assuming that the location of the pinching antenna is fixed, problem \eqref{pb:5} can be decomposed into the following $M$ decoupled subproblems:
   \begin{problem}\label{pb:6} 
  \begin{alignat}{2}
 \underset{ \mathcal{S}_P^m }{\rm{min}}    &\quad     P_{0m}+P_{mm}+P_{m} 
\\ s.t. &\quad    \eqref{pb:5 con2}, \eqref{pb:5 con3},
  \end{alignat}
\end{problem}
where $\mathcal{S}_P^m=\{P_{0m},P_{mm},P_{m} \}$.
Because the constraint $R^m_{0m}\geq R$ is removed, the optimization of ${\rm U}_{0m}$'s signal power is decoupled from that of  ${\rm U}_{m}$'s signal powers. In particular, the optimal solution of $P_{0m}$ can be obtained straightforwardly as follows: $P_{0m}^*=\epsilon  \left| {\boldsymbol \psi} _{0m}- {\boldsymbol \psi}^{\rm Pin}\right|^2$. Therefore, ${\rm U}_{m}$'s signal powers can be optimized as follows: 
 \begin{problem}\label{pb:7} 
  \begin{alignat}{2}
 \underset{ P_{mm},P_m}{\rm{min}}    &\quad    P_{mm}+P_{m}
\\ \label{st1 P7} s.t. &\quad   
\log_2\left(
1+ \frac{  P_{mm} g_{mm}}
{     P_{0m}g_{mm}+1}
\right) +
\log_2\left(
1+    P_m g_m   
\right)\geq R,
  \end{alignat}
\end{problem}
where $g_m=\frac{\eta }{P_{\rm N} \left| {\boldsymbol \psi} _m - {\boldsymbol \psi}_m^{\rm BS}\right|^2}  
$ and $g_{mm}=\frac{\eta }{P_{\rm N} \left| {\boldsymbol \psi} _m - {\boldsymbol \psi}^{\rm Pin}\right|^2}  
$. 

It is straightforward to show that problem \eqref{pb:7} is a convex optimization problem with respect to  $P_{mm}$ and $P_{m}$. The Lagrangian of problem \eqref{pb:7} is given by
\begin{align}\label{lagp1}
\mathcal{L} =& P_{mm}+P_{m}+\lambda_0 \left(R-
\log_2\left(
1+    P_m g_m   
\right)\right.\\\nonumber &\left.-
\log_2\left(
1+ \frac{  P_{mm} g_{mm}}
{     P_{0m}g_{mm}+1}
\right) \right)-\lambda_1P_{mm} -\lambda_2P_{m},
\end{align}
where $\lambda_m$, $0\leq m \leq 2$, denote the Lagrangian multiplers corresponding to \eqref{st1 P7}, $P_{mm}\geq 0$ and $P_m\geq 0$, respectively. By using the Lagrangian shown in \eqref{lagp1} and also applying the Karush–Kuhn–Tucker (KKT) conditions, a closed-form expression for the optimal transmit power solution of problem \eqref{pb:7} can be obtained \cite{Boyd}. 

We note that problem \eqref{pb:7} is similar to the optimization problems obtained for power minimization in hybrid NOMA, and by following steps similar to those in \cite{9679390},   the optimal solution for  $P_{mm}$ can be obtained as follows:
\begin{align}
P_{mm} &=
\begin{cases}
\displaystyle
\sqrt{\frac{2^R \left(P_{0m}g_{mm}+1\right)}{g_mg_{mm}}}
-\frac{P_{0m}g_{mm}+1}{g_{mm}}, & {\rm Case \, I}\\[2ex]
0, & {\rm Case \, II}\\[1ex]
\displaystyle
\frac{(2^R-1)\left(P_{0m}g_{mm}+1\right)}{g_{mm}}, & {\rm Case \, III}
\end{cases},
\end{align}
and $P_{m}$ can be obtained as follows:
\begin{align}
P_m &=
\begin{cases}
\displaystyle
\sqrt{\frac{2^R \left(P_{0m}g_{mm}+1\right)}{g_mg_{mm}}}
-\frac{1}{g_m}, & {\rm Case \, I}\\[2ex]
\displaystyle
\frac{2^R-1}{g_m}, & {\rm Case \, II}\\[2ex]
0, & {\rm Case \, III}
\end{cases},
\end{align}
where the three cases are given by
\begin{align}
\begin{cases}
\displaystyle
{\rm Case \, I:}\quad &
\frac{1}{2^Rg_m}<  \frac{P_{0m}g_{mm}+1}{g_{mm}}< \frac{2^R}{g_m} 
\\
{\rm Case \, II:}\quad &
  \frac{P_{0m}g_{mm}+1}{g_{mm}}\geq\frac{2^R}{g_m}
\\\label{case iii}
{\rm Case \, III:}\quad &
 {\frac{2^R\left(
P_{0m}g_{mm}+1
\right)}{  g_{mm}}} \leq \frac{1}{g_m}
\end{cases}.
\end{align}

Recall that the case of complete offloading, i.e., Case III, is of interest. By combining \eqref{case iii} with the fact that $P_{0m}=\epsilon  \left| {\boldsymbol \psi} _{0m}- {\boldsymbol \psi}^{\rm Pin}\right|^2$, the optimality condition for complete offloading is obtained as follows: 
\begin{align}
\epsilon  \left| {\boldsymbol \psi} _{0m}- {\boldsymbol \psi}^{\rm Pin}\right|^2+ \frac{  
 1
 }{  g_{mm}} \leq \frac{1}{2^R g_m},
\end{align}
which can be further expressed as the following explicit function of the users' locations:
\begin{align}\label{hidden1}
 & \left| {\boldsymbol \psi} _m - {\boldsymbol \psi}_m^{\rm BS}\right|^2\geq  2^R\\\nonumber &\times\left( \left(2^R-1\right)   \left| {\boldsymbol \psi} _{0m}- {\boldsymbol \psi}^{\rm Pin}\right|^2+   \left| {\boldsymbol \psi} _m - {\boldsymbol \psi}^{\rm Pin}\right|^2\right)   .
\end{align}

{\it Remark 2:}  The result in \eqref{hidden1} defines the optimality condition of complete offloading, which can be explained as follows. In particular, if the distance between ${\rm U}_m$  and ${\rm BS}_m$ is much larger than the distance between ${\rm U}_m$ and the pinching antenna, it is intuitive to completely offload the user to ${\rm Cell}_0$. We note that \eqref{hidden1} imposes an additional constraint when optimizing the location of the pinching antenna. 

\subsubsection{Optimizing the pinching antenna location}
We note that there may be many feasible pinching antenna locations that ensure the optimality of complete traffic offloading, and the aim of this subsection is to find the optimal pinching antenna location that not only ensures the optimality of complete offloading but also minimizes the overall transmit power.

Since complete offloading is of interest,  $P_m=0$, and hence problem \eqref{pb:5} can be expressed as follows:
   \begin{problem}\label{pb:8} 
  \begin{alignat}{2}
 \underset{  {\boldsymbol \psi}^{\rm Pin}}{\rm{min}}  &\quad   \sum_{m=1}^{M} \left(P_{0m}+P_{mm} \right)
\\ s.t. &\quad   \nonumber
 \left| {\boldsymbol \psi} _m - {\boldsymbol \psi}_m^{\rm BS}\right|^2\geq 2^R \left( \left(2^R-1\right)   \left| {\boldsymbol \psi} _{0m}- {\boldsymbol \psi}^{\rm Pin}\right|^2  \right.\\ \label{st8:1} &\quad  +\left. \left| {\boldsymbol \psi} _m - {\boldsymbol \psi}^{\rm Pin}\right|^2\right),\quad 1\leq m \leq M, \\\label{st8:2} &\quad  \left| {\boldsymbol \psi} _{0m} - {\boldsymbol \psi}^{\rm Pin}\right|^2\leq \left| {\boldsymbol \psi} _m - {\boldsymbol \psi}^{\rm Pin}\right|^2, \quad 1\leq m \leq M.
  \end{alignat}
\end{problem}

To simplify the objective function of problem \eqref{pb:8}, the closed-form solution for $P_{mm}$ needs to be rewritten as follows:
\begin{align} 
P_{mm} =&\frac{(2^R-1)\left(P_{0m}g_{mm}+1\right)}{g_{mm}}\\\nonumber =&(2^R-1) \epsilon  \left| {\boldsymbol \psi} _{0m}- {\boldsymbol \psi}^{\rm Pin}\right|^2  +  \epsilon \left| {\boldsymbol \psi} _m - {\boldsymbol \psi}^{\rm Pin}\right|^2 .\label{pmmxx}
\end{align}

By using the above expression of  $P_{mm}$ and also the closed-form expression of $P_{0m}$, problem \eqref{pb:8} can be expressed as follows:
   \begin{problem}\label{pb:9} 
  \begin{alignat}{2}
 \underset{ x^{\rm Pin}}{\rm{min}}  &\quad   \sum_{m=1}^{M} \left(  2^R   \left(x^{\rm Pin}-x_{0m}\right)^2  +   \left(x^{\rm Pin}-x_{m}\right)^2    \right)
\\ \nonumber s.t. &\quad  \left(2^R-1\right)  \left(x^{\rm Pin}-x_{0m}\right)^2 +  \left(x^{\rm Pin}-x_{m}\right)^2\leq d_{1m}, \\\label{const 9 :9} &\quad\quad\quad\quad\quad\quad\quad\quad\quad\quad\quad \quad 1\leq m \leq M, \\  &\quad \label{const 9 :10}
2\left(
x_{m}-x_{0m}
\right)  x^{\rm Pin} \leq   d_{2m}   , \quad 1\leq m \leq M.
  \end{alignat}
\end{problem}
where $d_m= 2^{-R}
 \left| {\boldsymbol \psi} _m - {\boldsymbol \psi}_m^{\rm BS}\right|^2- \left(2^R-1\right)\left(y^{\rm Pin}-y_{0m}\right)^2 -\left(2^R-1\right)d^2   - \left(y^{\rm Pin}-y_{m}\right)^2 + d^2 $ and $d_{2m}= x_{m} ^2+\left(y^{\rm Pin}-y_{m}\right)^2-x_{0m}^2-\left(y^{\rm Pin}-y_{0m}\right)^2$.

{\it Remark 3:} Problem \eqref{pb:9} is a convex optimization problem, and hence can be efficiently solved by off-the-shelf solvers. We note that constraint  \eqref{const 9 :9} is due to the fact that complete offloading is focused on, and constraint  \eqref{const 9 :10} is due to the adopted SIC decoding order.  An interesting observation is that if $R\geq \log_2\phi$, where $\phi$ is the golden ratio, i.e., $\phi=\frac{1+\sqrt{5}}{2}$, constraint  \eqref{const 9 :9} guarantees constraint  \eqref{const 9 :10}, as explained in the following. Constraint \eqref{const 9 :9} (or equivalently \eqref{st8:1}) can be rewritten as follows: 
\begin{align}\label{cost343}
2^R\left(2^R-1\right)   \left| {\boldsymbol \psi} _{0m}- {\boldsymbol \psi}^{\rm Pin}\right|^2   \leq &
 \left| {\boldsymbol \psi} _m - {\boldsymbol \psi}_m^{\rm BS}\right|^2\\\nonumber &-2^R  \left| {\boldsymbol \psi} _m - {\boldsymbol \psi}^{\rm Pin}\right|^2.
\end{align}
It is straightforward to show that if $R\geq \log_2\phi$, $2^R\left(2^R-1\right) \geq 1$, which means that any solution of ${\boldsymbol \psi}^{\rm Pin}$ which satisfies \eqref{cost343} also satisfies \eqref{sicconstraint1} (or equivallently \eqref{st8:2}).

   \vspace{-0.5em}
\subsection{${\rm U}_{m}$ Carries Out SIC} 
Recall that on ${\rm U}_{0m}$'s subcarrier, two downlink users, ${\rm U}_{0m}$ and ${\rm U}_{m}$, are served simultaneously, where the base station broadcasts the mixture of the signals for ${\rm U}_{0m}$ and ${\rm U}_{m}$. Unlike the previous section, in this section, we focus on the case with  ${\rm U}_{m}$ carrying out SIC. 

Because ${\rm U}_{0m}$ does not carry out SIC, it decodes its signal directly by treating ${\rm U}_{m}$'s signal as noise, which means that ${\rm U}_{0m}$'s achievable data rate is given by 
\begin{align}
R_{0m}=\log_2\left(
1+ \frac{  \frac{\eta P_{0m}}{ \left| {\boldsymbol \psi} _{0m} - {\boldsymbol \psi}^{\rm Pin}\right|^2}}
{     \frac{\eta P_{mm}}{ \left| {\boldsymbol \psi} _{0m}- {\boldsymbol \psi}^{\rm Pin}\right|^2}+P_{\rm N}}
\right). 
\end{align}

${\rm U}_{m}$ can first decode ${\rm U}_{0m}$'s signal with the following achievable data rate:
\begin{align}
R_{0m}^m=\log_2\left(
1+ \frac{  \frac{\eta P_{0m}}{ \left| {\boldsymbol \psi} _{m} - {\boldsymbol \psi}^{\rm Pin}\right|^2}}
{     \frac{\eta P_{mm}}{ \left| {\boldsymbol \psi} _{m}- {\boldsymbol \psi}^{\rm Pin}\right|^2}+P_{\rm N}}
\right). 
\end{align}
Assuming that ${\rm U}_{m}$ can decode ${\rm U}_{0m}$'s signal correctly, ${\rm U}_{m}$ can subtract ${\rm U}_{0m}$'s signal and then decode its own signal with the following achievable data rate on  ${\rm U}_{0m}$'s subcarrier: $\log_2\left(
1+   \frac{\eta P_{mm}}{P_{\rm N} \left| {\boldsymbol \psi} _m - {\boldsymbol \psi}^{\rm Pin}\right|^2} 
\right)$. Since ${\rm U}_{m}$ receives two data streams from both ${\rm BS}_0$ and ${\rm BS}_m$,   ${\rm U}_{m}$'s overall achievable data rate is given by
\begin{align}
R_m=&\log_2\left(
1+   \frac{\eta P_{mm}}{P_{\rm N} \left| {\boldsymbol \psi} _m - {\boldsymbol \psi}^{\rm Pin}\right|^2} 
\right) \\\nonumber &+\log_2\left(
1+   \frac{\eta P_{m}}{P_{\rm N} \left| {\boldsymbol \psi} _m - {\boldsymbol \psi}_m^{\rm BS}\right|^2} 
\right) . 
\end{align}

Therefore, the considered total transmit power minimization problem can be formulated as follows:
 \begin{problem}\label{pb:12} 
  \begin{alignat}{2}
 \underset{ \mathcal{S}_P, {\boldsymbol \psi}^{\rm Pin}}{\rm{min}} &\quad   \sum_{m=1}^{M} \left(P_{0m}+P_{mm}+P_{m}\right)
\\ s.t. &\quad   
\log_2\left(
1+ \frac{  \frac{\eta P_{0m}}{ \left| {\boldsymbol \psi} _{0m} - {\boldsymbol \psi}^{\rm Pin}\right|^2}}
{     \frac{\eta P_{mm}}{ \left| {\boldsymbol \psi} _{0m}- {\boldsymbol \psi}^{\rm Pin}\right|^2}+P_{\rm N}}
\right) \geq R, \quad 1\leq m \leq M,\\  &\quad
\log_2\left(
1+   \frac{\eta P_{mm}}{P_{\rm N} \left| {\boldsymbol \psi} _m - {\boldsymbol \psi}^{\rm Pin}\right|^2} 
\right) \\\nonumber &+\log_2\left(
1+   \frac{\eta P_{m}}{P_{\rm N} \left| {\boldsymbol \psi} _m - {\boldsymbol \psi}_m^{\rm BS}\right|^2} 
\right) \geq R,\quad 1\leq m \leq M,\\ &\quad  \left| {\boldsymbol \psi} _{0m} - {\boldsymbol \psi}^{\rm Pin}\right|^2\geq  \left| {\boldsymbol \psi} _m - {\boldsymbol \psi}^{\rm Pin}\right|^2, \quad 1\leq m \leq M.\label{pb:12 3} 
  \end{alignat}
\end{problem}

{\it Remark 4:} Similar to the case considered in the previous section, it is assumed that ${\rm U}_m$ can carry out SIC successfully, i.e., $R_{0m}^m\geq R_{0m}$, which leads to the constraint shown in \eqref{pb:12 3}. Similar to problem \eqref{pb:5}, problem \eqref{pb:12} can be solved in  two stages, as shown in the following subsections. 
 
 \subsubsection{Optimizing the transmit power for a given antenna location}
 For a given pinching antenna location, problem \eqref{pb:12} can be decomposed into the following $M$ decoupled subproblems: 
  \begin{problem}\label{pb:13} 
  \begin{alignat}{2}
 \underset{ \mathcal{S}_P^m}{\rm{min}}  &\quad     P_{0m}+P_{mm}+P_{m} 
\\ s.t. &\quad    
 \left(2^R-1\right)   g_{0m}P_{mm} - g_{0m}P_{0m}  +\left(2^R-1\right)
 \leq 0, \\  &\quad
R-\log_2\left(
1+   g_{mm}  P_{mm}  
\right) -\log_2\left(
1+  g_m P_m 
\right) \leq 0 ,
  \end{alignat}
\end{problem}
 where  $g_{0m}=\frac{\eta }{P_{\rm N} \left| {\boldsymbol \psi} _{0m} - {\boldsymbol \psi}^{\rm Pin}\right|^2}  
$. 
 
 Unlike problem \eqref{pb:5} or the problems in \cite{9679390}, for problem \eqref{pb:13}, the optimization of ${\rm U}_{0m}$'s transmit power, $P_m$, is inherently coupled with that of $P_m$ and $P_{mm}$, and therefore it cannot be solved separately. Nevertheless, the optimal transmit powers can be obtained in closed-form as shown in the following lemma. 
 \begin{lemma}\label{lemma2}
 The optimal solution of problem \eqref{pb:13} can be expressed as follows:
\begin{align}
P_{mm} &=
\begin{cases}
\displaystyle
 \sqrt{\frac{1}{g_{mm}g_m}}-\frac{1}{g_{mm}}, & {\rm Case \, I}\\[2ex]
0, & {\rm Case \, II}\\[1ex]
\displaystyle
\frac{2^{R}-1}{g_{mm}}, & {\rm Case \, III}
\end{cases},
\end{align}
\begin{align}
P_{0m} &=
\begin{cases}
\displaystyle
 \left(2^R-1\right)P_{mm}+\frac{2^R-1}{g_{0m}}, & {\rm Case \, I}\\[2ex]
\displaystyle
\frac{2^R-1}{g_{0m}}, & {\rm Case \, II}\\[2ex]
   \frac{ \left(2^R-1\right)^2 }{g_{mm}}   +\frac{2^R-1}{ g_{0m}} , & {\rm Case \, III}
\end{cases},
\end{align} 
\begin{align}
P_m &=
\begin{cases}
\displaystyle
\frac{ 2^R}{\sqrt{g_mg_{mm}}}-\frac{1}{g_m}, & {\rm Case \, I}\\[2ex]
\displaystyle
\frac{2^R-1}{g_m}, & {\rm Case \, II}\\[2ex]
0, & {\rm Case \, III}
\end{cases},
\end{align}
where the conditions for the three cases are given by
\begin{align}
\begin{cases}
\displaystyle
{\rm Case \, I:}\quad &
 1< \frac{g_{mm}}{g_m} < 2^{2R}
\\
{\rm Case \, II:}\quad &
  \frac{ g_{mm} }{g_m }  \leq 1
\\\label{case iii ?}
{\rm Case \, III:}\quad &
 \frac{g_{mm}}{g_m}\geq  2^{2R}
 \end{cases}.
\end{align}
\end{lemma}
\begin{proof}
See Appendix \ref{proof2}.
\end{proof}
 
 \subsubsection{Optimizing the pinching antenna location}
Similar to the previous section, the case with complete traffic offloading will be focused on here, since it leads to a reduction of the traffic and an increase of the available bandwidth in ${\rm Cell}_m$, $1\leq m \leq M$. 

By using the optimality condition for Case III shown in Lemma \ref{lemma2},   problem \eqref{pb:12} can be expressed as follows:
   \begin{problem}\label{pb:14} 
  \begin{alignat}{2}
 \underset{   {\boldsymbol \psi}^{\rm Pin}}{\rm{min}}  &\quad   \sum_{m=1}^{M} \left(P_{0m}+P_{mm} \right)
\\ \label{st14:1}s.t. &\quad  
   \left| {\boldsymbol \psi} _m - {\boldsymbol \psi}_m^{\rm BS}\right|^2  \geq  2^{2R}   \left| {\boldsymbol \psi} _m - {\boldsymbol \psi}^{\rm Pin}\right|^2, \quad 1\leq m \leq M,\\ &\quad \label{st14:2}
    \left| {\boldsymbol \psi} _{0m} - {\boldsymbol \psi}^{\rm Pin}\right|^2\geq  \left| {\boldsymbol \psi} _m - {\boldsymbol \psi}^{\rm Pin}\right|^2,  \quad 1\leq m \leq M ,
  \end{alignat}
\end{problem}
where the constraint in \eqref{st14:1} is due to the fact that the complete offloading case is focused on, and the constraint in \eqref{st14:2} is due to the adopted SIC decoding order. 
 
By using the closed-form expressions in Lemma \ref{lemma2},  the objective function of problem \eqref{pb:14} can be expressed as the following explicit function of the pinching antenna location:
\begin{align} \label{total3}
P_{0m}+P_{mm} =&   \frac{2^R \left(2^R-1\right) }{g_{mm}}   +\frac{2^R-1}{ g_{0m}} 
\\\nonumber=& 
 2^R\epsilon   \left| {\boldsymbol \psi} _m - {\boldsymbol \psi}^{\rm Pin}\right|^2   + \epsilon \left| {\boldsymbol \psi} _{0m} - {\boldsymbol \psi}^{\rm Pin}\right|^2  .
\end{align}

By using the expression for the total transmit power shown in \eqref{total3}, problem \eqref{pb:14} can be expressed as follows:
   \begin{problem}\label{pb:15} 
  \begin{alignat}{2}
 \underset{  {\boldsymbol \psi}^{\rm Pin}}{\rm{min}}   &\quad   \sum_{m=1}^{M} \left(   2^R\epsilon   \left| {\boldsymbol \psi} _m - {\boldsymbol \psi}^{\rm Pin}\right|^2   + \epsilon \left| {\boldsymbol \psi} _{0m} - {\boldsymbol \psi}^{\rm Pin}\right|^2  \right)
\\ \label{const 15 :9}s.t. &\quad  
\eqref{st14:1}, \eqref{st14:2}. 
  \end{alignat}
\end{problem}
 
Problem \eqref{pb:15} can be further expressed as the following convex form:
 \begin{problem}\label{pb:16} 
  \begin{alignat}{2}
 \underset{ x^{\rm Pin}}{\rm{min}} &\quad   \sum_{m=1}^{M} \left(   2^R   \left(
 x^{\rm Pin}-x_m
 \right)^2   +    \left(
 x^{\rm Pin}-x_{0m}
 \right)^2   \right)
\\ \label{const 16 :1}s.t. &\quad  
      2^{2R}  \left(
 x^{\rm Pin}-x_m
 \right)^2 \leq   d_{3m},\quad 1\leq m \leq M,\\ \label{const 16 :2}&\quad 
2   \left(
x_m-x_{0m}
 \right) x^{\rm Pin}    >    d_{2m}  ,  \quad 1\leq m \leq M ,
  \end{alignat}
\end{problem}
where $d_{3m}= \left| {\boldsymbol \psi} _m - {\boldsymbol \psi}_m^{\rm BS}\right|^2 -2^{2R} \left(
 y^{\rm Pin}-y_m
 \right)^2- 2^{2R} d^2$. 
It is straightforward to verify that problem \eqref{pb:16} is a convex optimization problem, and hence it can be solved efficiently by off-the-shelf optimization solvers. 

After the optimal solutions for problems \eqref{pb:9} and \eqref{pb:14} are obtained, the minimal transmit power can be determined by comparing the objective values of the two problems. 
 
 \vspace{-0.5em}
 
\subsection{A Special Case Study}
This subsection focuses on the special case with two cells, i.e., ${\rm BS}_1$ offloads ${\rm U}_1$'s traffic to ${\rm BS}_0$. Recall that two optimization problems have been formulated for the cases with different SIC decoding orders. As can be seen from \eqref{pb:8} and \eqref{pb:14}, there are two constraints for each of the formulated problems, one for the optimality of the complete offloading strategy and the other due to the adopted SIC decoding order. In this section, we assume that the distance between ${\rm U}_1$ and ${\rm BS}_1$ is so large that the constraints in \eqref{st8:1} and \eqref{st14:1} always hold for any choice of $x^{\rm Pin}$, i.e., complete offloading is always preferred. For example, if ${\rm BS}_1$ is a macro base station, it is very likely that the distance between ${\rm U}_1$ and ${\rm BS}_1$ is much larger than the distances between the user and the pinching antenna of the small cell base station. 

With the aforementioned assumptions, a concise closed-form expression can be obtained for the optimal antenna location. In particular, denote the strong and weak users by ${\rm U}_{\rm s}$ and ${\rm U}_{\rm w}$, respectively. In particular, if $|y_{01}|\leq |y_1|$, ${\rm U}_{01}$ is the strong user, and ${\rm U}_{1}$ is the weak user, i.e., ${\rm U}_{\rm s}={\rm U}_{01}$ and ${\rm U}_{\rm w}={\rm U}_{1}$, otherwise ${\rm U}_{\rm s}={\rm U}_{1}$ and ${\rm U}_{\rm w}={\rm U}_{01}$. The locations of ${\rm U}_{\rm s}$ and ${\rm U}_{\rm w}$ are denoted by $(x_{\rm s}, y_{\rm s},0)$ and $(x_{\rm w}, y_{\rm w},0)$, respectively. For the waveguide deployment illustrated in Fig. \ref{fig1b}, we assume that $y^{\rm Pin}=0$. With the above assumptions, the optimal pinching antenna location and the minimal transmit power are provided in the following lemma.

      \begin{figure}[t]\centering \vspace{-0.2em}
    \epsfig{file=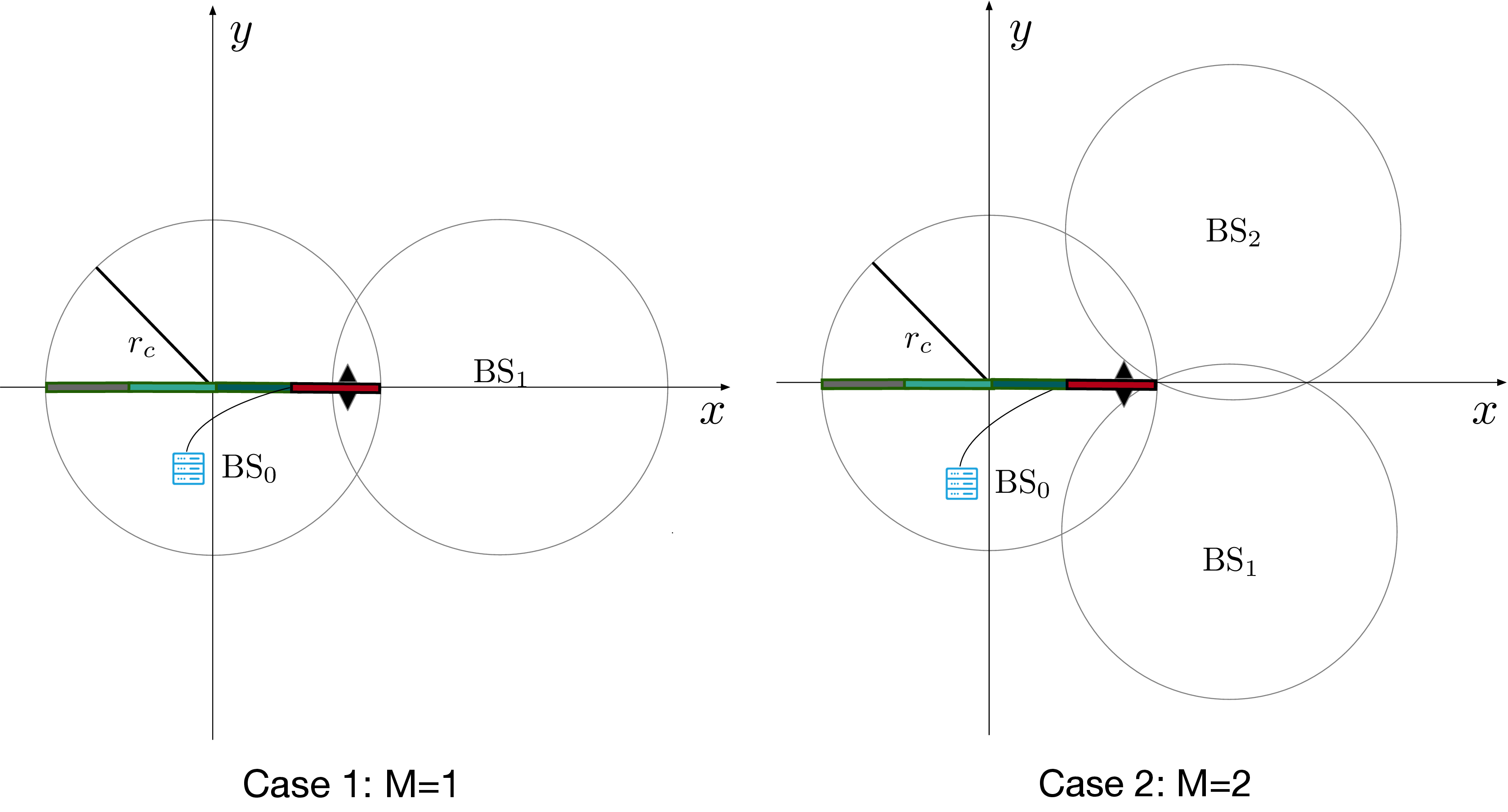, width=0.4\textwidth, clip=}\vspace{-0.5em}
\caption{Illustration of the two considered network topologies. In Case I, there are two cells, i.e., $M=1$, where both ${\rm U}_0$ and ${\rm U}_1$ are uniformly distributed in the overlapping area between the two cells. In Case II, there are three cells, i.e., ${M=2}$, where ${\rm U}_m$ is uniformly distributed in the overlapping area between ${\rm Cell}_0$ and ${\rm Cell}_m$, $m\in\{1,2\}$. For the strategy considered in Section \ref{section III}, there is a single user (${\rm U}_0$) served by the pinching antenna in ${\rm Cell}_0$, and the location of ${\rm U}_0$ is generated in the same manner as in Case I. For the strategy considered in Section \ref{section IV}, each offloaded user, ${\rm U}_m$, is paired with an existing user (${\rm U}_{0m}$), where ${\rm U}_{0m}$ and ${\rm U}_{m}$ are uniformly distributed in the same overlapping area.  
  \vspace{-1em}    }\label{figsim0}   \vspace{-1em} 
\end{figure}

\begin{lemma}\label{lemma3}
For the considered special case, the minimal transmit power can be achieved by asking ${\rm U}_{\rm s}$ to carry out SIC, and the minimal transmit power is given by
\begin{align}\label{lemma3mi}
P^{\rm all*} =&   
  \frac{\epsilon 2^R}{2^R+1}    \left(    x_{\rm s}- x_{\rm w}  \right)^2+ \epsilon\left(2^R y_{\rm s}^2 +y_{\rm w}^2\right) +\epsilon(2^R+1)d^2,
\end{align} and  the optimal  pinching antenna location is given by
\begin{align}\label{xpinlemma2}
x^{\rm Pin *}  =
\frac{  \left(  2^R     x_{\rm s}+ x_{
\rm w}\right) }{ 2^R+1 } . 
\end{align} 

\end{lemma}
\begin{proof}
See Appendix \ref{proof3}.
\end{proof}
{\it Remark 5:} Lemma \ref{lemma3} shows that, for the considered two-cell special case, multiple resource allocation problems based on different SIC decoding orders are not needed. In particular, there is a dominant SIC order, which leads to the minimal overall transmit power. Furthermore, the closed-form solution shown in \eqref{xpinlemma2} reveals that the optimal antenna location is a function of the projections of the users' positions onto the waveguide only. However, we note that the distances between the users and the waveguide, i.e., $y_m$ and $y_{0m}$, still affect the optimal antenna location by determining which SIC decoding order is dominant.

%
%

\section{Numerical Results}\label{section V}
 In this section, the performance of pinching-antenna-assisted traffic offloading is evaluated by using computer simulations. The network topologies shown in Fig. \ref{figsim0} are used for the conducted simulations. Throughout this section, it is assumed that  $f_c=28$ GHz, $d=3$ m, $N=4$, and the noise power is $-70$ dBm. The cell radius is denoted by $r_c$, and the distance between ${\rm BS}_0$ and ${\rm BS}_m$, $1\leq m \leq M$, is denoted by $r_d=\sqrt{3}r_c$\footnote{The choice of $r_d=\sqrt{3}r_c$ ensures that there is a single intersection between the three circles in Fig. \ref{figsim0}.}. The waveguide in ${\rm Cell}_0$ is placed on the $x$ axis of the plane shown in Fig. \ref{figsim0}, and its length is $2r_c$. Furthermore, as illustrated in Fig. \ref{fig1b}, the rightmost waveguide segment is the focus of interest, and hence $ x_m\geq \frac{N-2}{N}r_c$, $0\leq m \leq M$, i.e., the scheduled users should not fall to the left-hand side of the considered segment.

   \begin{figure}[!] \vspace{-2em}
\begin{center}
\subfigure[ Case I with $r_c=40$ m ]{\label{figsim1a}\includegraphics[width=0.321\textwidth]{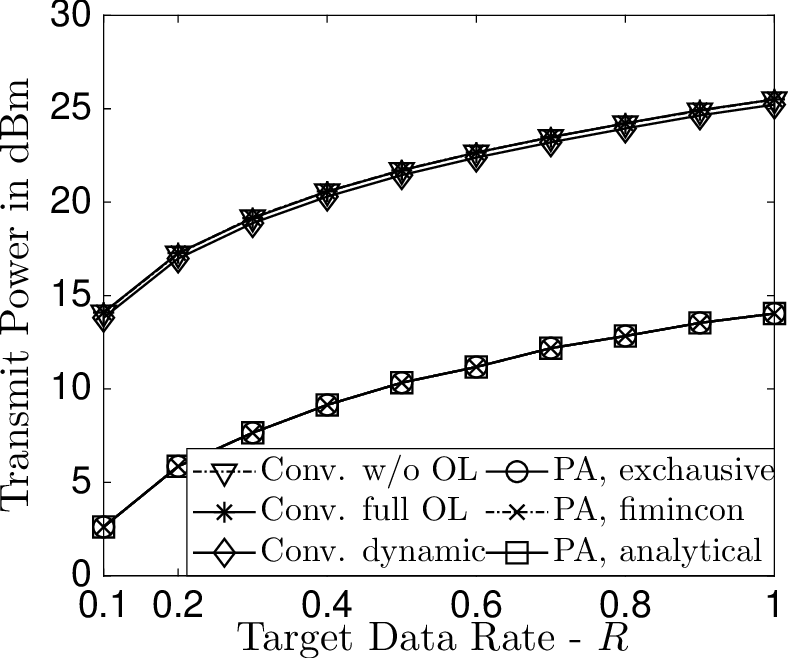}} 
\subfigure[Case I with  $r_c=80$ m ]{\label{figsim1b}\includegraphics[width=0.321\textwidth]{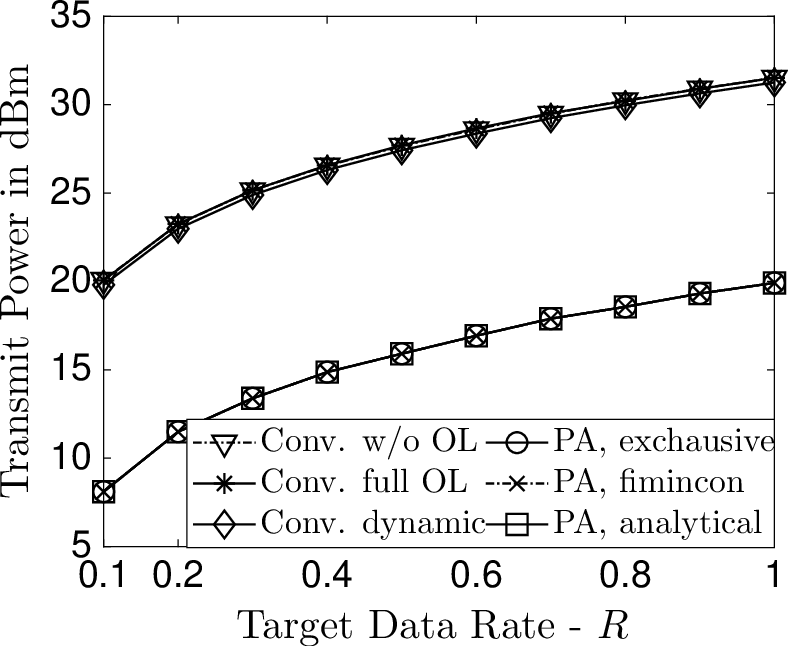}}   \vspace{-1em}
\end{center}
\caption{Impact of pinching antennas on the transmit power for the considered traffic offloading schemes, where Case I is focused on and the scenario in Section \ref{section III} is considered.     \vspace{-1em} }\label{figsim1}\vspace{-1.2em}
\end{figure}

Figs. \ref{figsim1}, \ref{figsim2}, and \ref{figsim3} focus on the case in which ${\rm U}_m$ keeps using its subcarrier in ${\rm Cell}_m$, i.e., the scenario considered in Section \ref{section III}. In particular, Fig.  \ref{figsim1} compares the transmit powers required by different traffic offloading schemes in order to achieve the target data rate $R$, where Case I shown in Fig. \ref{figsim0} is focused on. In order to facilitate the performance analysis, three types of conventional-antenna-based schemes are used for benchmarking, where the first one does not employ traffic offloading, the second one forces ${\rm BS}_0$ to serve ${\rm U}_m$, $m\in \{1, \cdots, M\}$, and the third one is based on dynamic traffic offloading to realize the minimum between the transmit powers achieved by the first two benchmarking schemes. As can be seen from the figure, the use of pinching antennas can significantly reduce the transmit power for traffic offloading. For example, for Case I with $r_c=40$ m, the performance gain of the pinching-antenna scheme over the conventional ones  is more than $10$ dBm, i.e., the use of pinching antennas can reduce the transmit power by more than a factor of $10$. For the optimization problem shown in \eqref{pb:3}, the optimal solution can be obtained by applying an exhaustive search, an off-the-shelf optimization solver, such as Matlab fmincon, as well as in a closed-form expression provided in \eqref{analyical 1} for the special case of $M=1$. Fig. \ref{figsim1} shows that the performance achieved with the analytical solution is identical to that obtained with fmincon and an exhaustive search, which verifies the optimality of the obtained analytical result shown in \eqref{analyical 1}.

   \begin{figure}[!] \vspace{-2em}
\begin{center}
\subfigure[ Case II with $r_c=40$ m ]{\label{figsim2a}\includegraphics[width=0.321\textwidth]{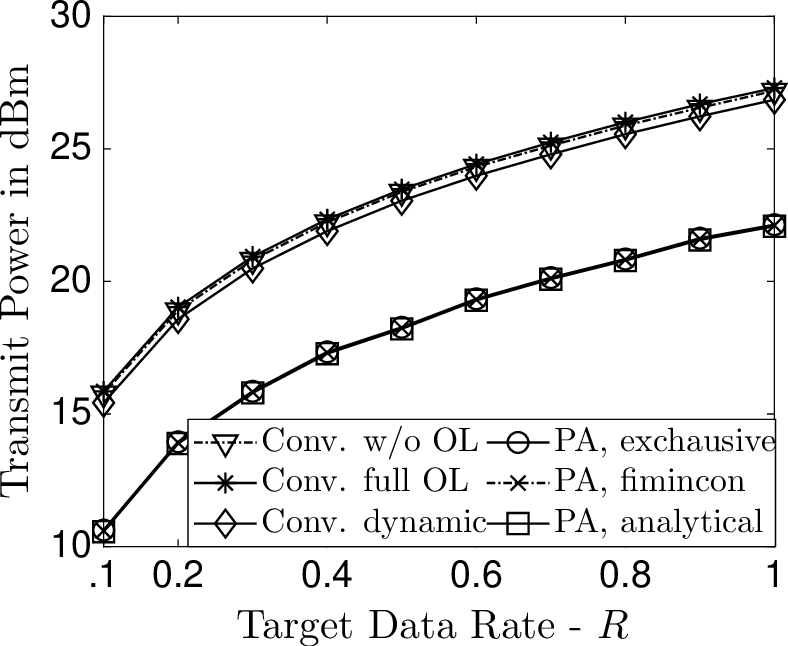}} 
\subfigure[Case II with $r_c=80$ m ]{\label{figsim2b}\includegraphics[width=0.321\textwidth]{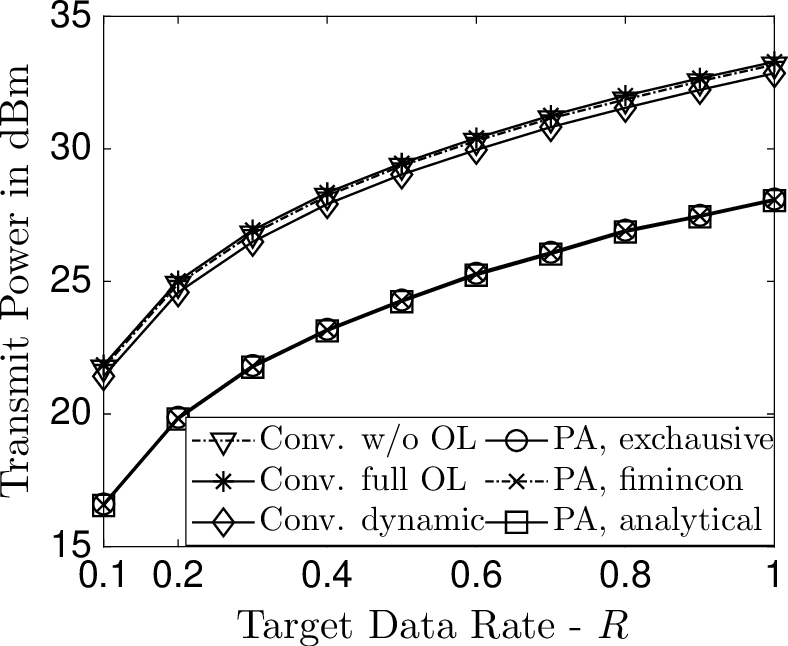}}   \vspace{-1em}
\end{center}
\caption{ Impact of pinching antennas on the transmit power for the considered traffic offloading schemes, where Case II is focused on and the scenario in Section \ref{section III} is considered.   \vspace{-1em} }\label{figsim2}\vspace{-1em}
\end{figure}

Fig. \ref{figsim2} focuses on Case II, i.e., there are three cells in the network and the users from ${\rm Cell}_1$ and ${\rm Cell}_2$ are offloaded and served by ${\rm BS}_0$. Consistent with Fig. \ref{figsim1}, Fig. \ref{figsim2} demonstrates that the use of pinching antennas can significantly reduce the transmit power required for traffic offloading. However, comparing Fig. \ref{figsim2} and Fig. \ref{figsim1}, one can observe that the performance gain of pinching-antenna systems over conventional-antenna systems decreases by increasing $M$. The reason for this performance degradation is that a single antenna is activated to serve multiple users from the $M+1$ cells. For the case with a small $M$, the location of the pinching antenna is tailored to the served users; however, with a large $M$, the adopted antenna location becomes a compromise choice, i.e., the adopted location might not be ideal for any individual user. 

   \begin{figure}[!] \vspace{-2em}
\begin{center}
\subfigure[ Conventional antennas ]{\label{figsim3a}\includegraphics[width=0.321\textwidth]{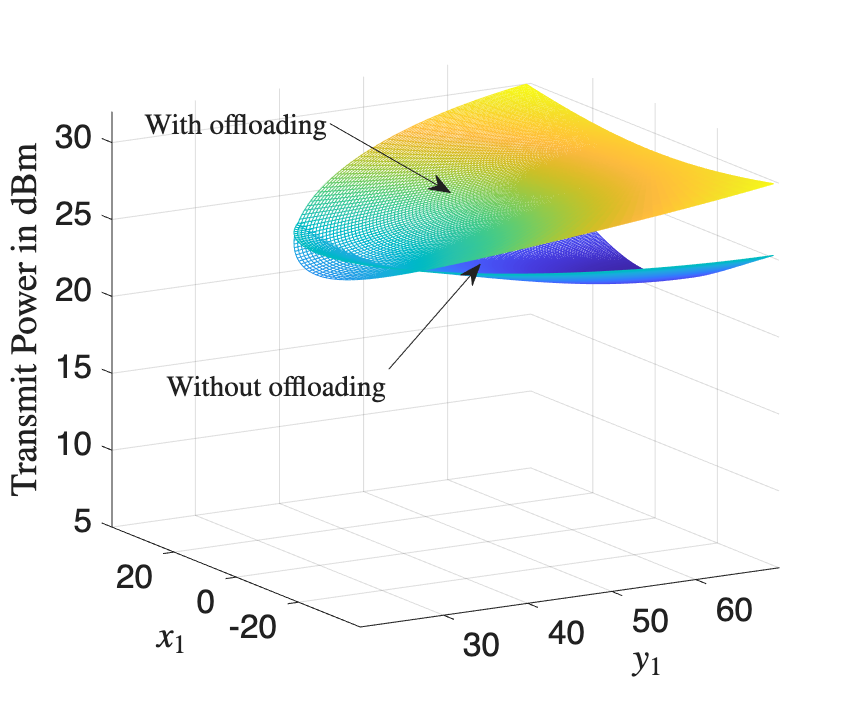}} 
\subfigure[Pinching antennas ]{\label{figsim3b}\includegraphics[width=0.321\textwidth]{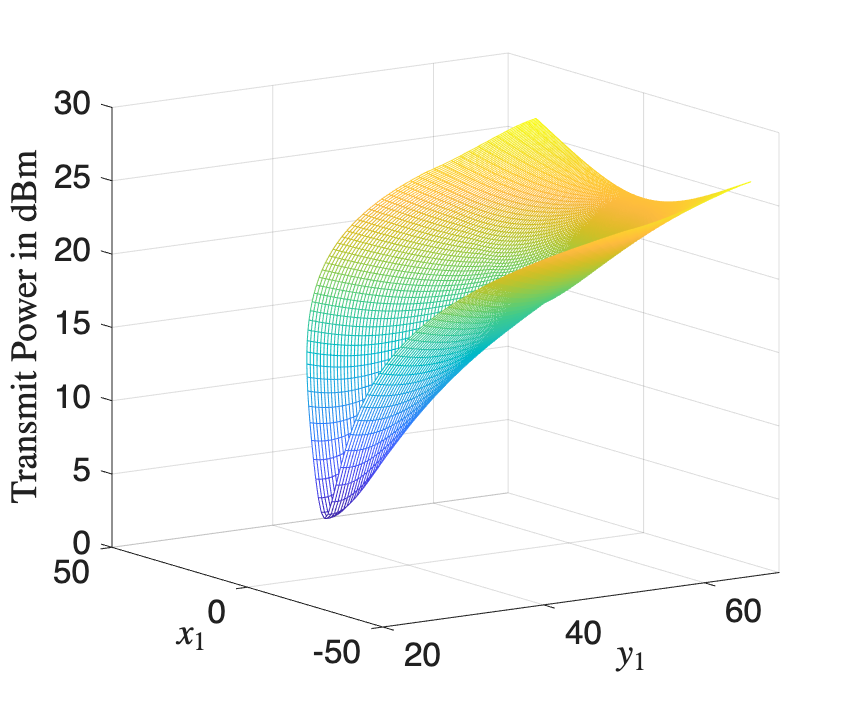}}   \vspace{-1em}
\end{center}
\caption{Impact of pinching antennas on the transmit power for the considered traffic offloading schemes, where Case I is focused on and the scenario in Section \ref{section III} is considered.  Unlike Fig. \ref{figsim1}, the location of ${\rm U}_0$ is fixed at the middle of the considered segment, i.e., $
\left(
\frac{N-1}{N}r_c, 0,0
\right)$, and ${\rm U}_1$ is assumed to be within a left-hand-side semicircular region of radius $r_c$ and its center at ${\rm BS}_1$, where $r_c=40$ m.    \vspace{-1em} }\label{figsim3}\vspace{-1em}
\end{figure}

Fig. \ref{figsim3} illustrates the performance gain of pinching-antenna systems over conventional-antenna systems with deterministic user locations. The key observation from Fig. \ref{figsim3} is that the performance gain of pinching antennas depends on ${\rm U}_1$'s location. In particular, if ${\rm U}_1$ is a cell edge user, the performance gain of pinching-antenna systems over conventional-antenna systems is enormous. The reason for this significant performance gain is that with conventional antennas, a cell-edge user is far away from both ${\rm BS}_0$ and ${\rm BS}_1$. Therefore, the user has to be served with a large transmit power, with or without traffic offloading. However, due to its flexibility, the pinching antenna can be placed close to the cell-edge user, which means that the user can served by ${\rm BS}_0$ with low transmit power. 

  \begin{figure}[t] \vspace{-2em}
\begin{center}
\subfigure[ Case I with $r_c=40$ m ]{\label{figsim4a}\includegraphics[width=0.321\textwidth]{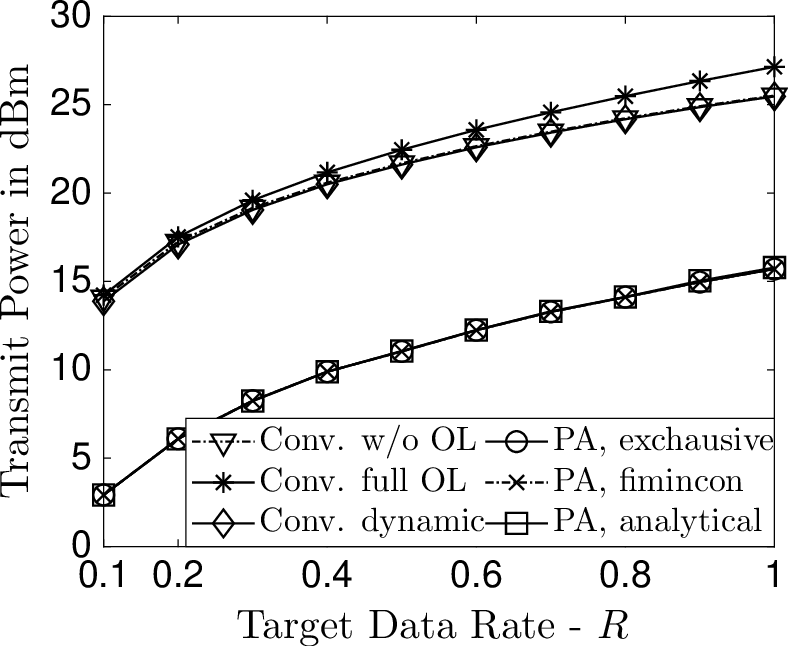}} 
\subfigure[Case I with $r_c=80$ m ]{\label{figsim4b}\includegraphics[width=0.321\textwidth]{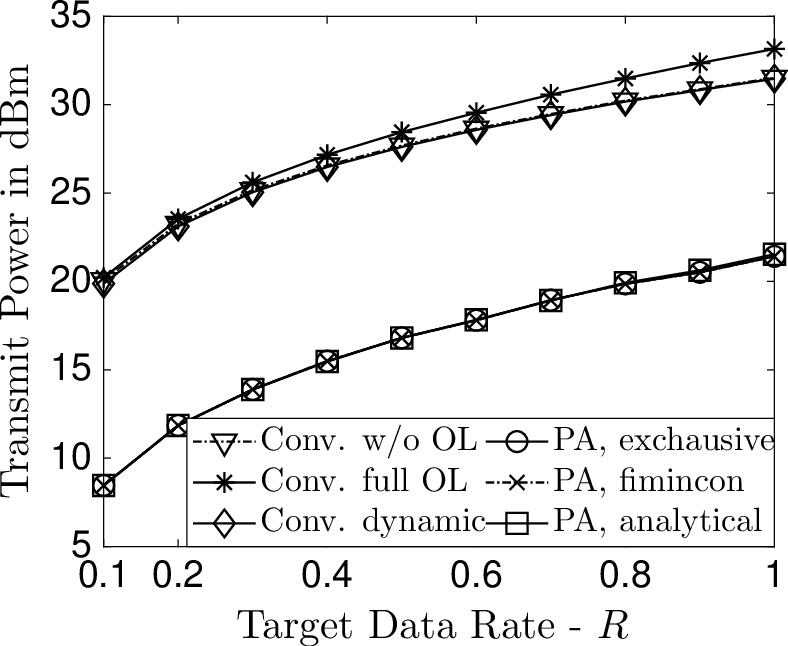}}   \vspace{-1em}
\end{center}
\caption{Impact of pinching antennas on the transmit power for the considered traffic offloading schemes, where Case I is focused on and the scenario in Section \ref{section IV} is considered.  \vspace{-1em} }\label{figsim4}\vspace{-1em}
\end{figure}
Figs. \ref{figsim4}, \ref{figsim5}, and \ref{figsim6} focus on the scenario in which ${\rm U}_m$ releases its subcarrier in ${\rm Cell}_m$, i.e., the scenario considered in Section \ref{section IV}. The two network topologies shown in Fig. \ref{figsim0} are still used.   Fig. \ref{figsim4} focuses on Case I, where there are two cells. Consistent with the previous results, Fig. \ref{figsim4} shows that using pinching antennas significantly reduces the transmit power, which is due to the fact that ${\rm BS}_0$ can place its pinching antenna close to the user to be offloaded from ${\rm Cell}_1$ to ${\rm Cell}_0$. We note that for the special case of $M=1$, closed-form expressions for the optimal solutions of power allocation and antenna placement can be obtained as shown in Lemma \ref{lemma3}. The fact that the analytical solution yields the same performance as the schemes based on an exhaustive search and optimization solvers verifies the optimality of the analytical solution. 

Fig. \ref{figsim5} focuses on Case II, where there are three cells. While pinching-antenna systems can still outperform the benchmarking schemes, the performance gain of pinching antennas over conventional antennas reduces by increasing $M$, a phenomenon consistent with the previously presented simulation results. Fig. \ref{figsim6} shows the performance gain of pinching-antenna systems over conventional-antenna systems for deterministic user locations. Similar to Fig. \ref{figsim3}, Fig. \ref{figsim6} shows that users close to the cell boundary benefit significantly from traffic offloading in pinching-antenna systems. Unlike Figs. \ref{figsim1}-\ref{figsim3}, Figs. \ref{figsim4}-\ref{figsim6} show that if conventional antennas are used, the scheme with complete traffic offloading suffers the worst performance, which means that with conventional antennas, traffic offloading almost never happens. However, the use of pinching antennas makes complete traffic offloading a preferred choice, which is important in practice. Take the scenario considered in Section \ref{section III} as an example. Complete traffic offloading is particularly meaningful, where the use of pinching antennas ensures that ${\rm U}_m$, $1\leq m\leq M$, can be offloaded to ${\rm Cell}_0$ and their bandwidth in ${\rm Cell}_m$ can be released to support additional users. In addition, pinching-antenna assisted traffic offloading requires lower transmit power and hence supports greener communications, compared to the benchmarking schemes.  

   \begin{figure}[t] \vspace{-2em}
\begin{center}
\subfigure[ Case II with $r_c=40$ m ]{\label{figsim5a}\includegraphics[width=0.321\textwidth]{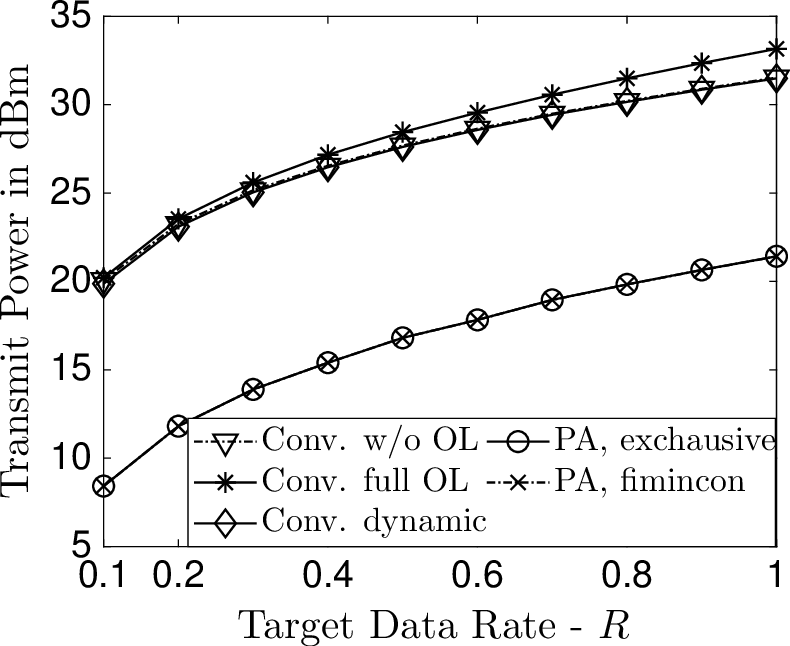}} 
\subfigure[Case II with $r_c=80$ m ]{\label{figsim5b}\includegraphics[width=0.321\textwidth]{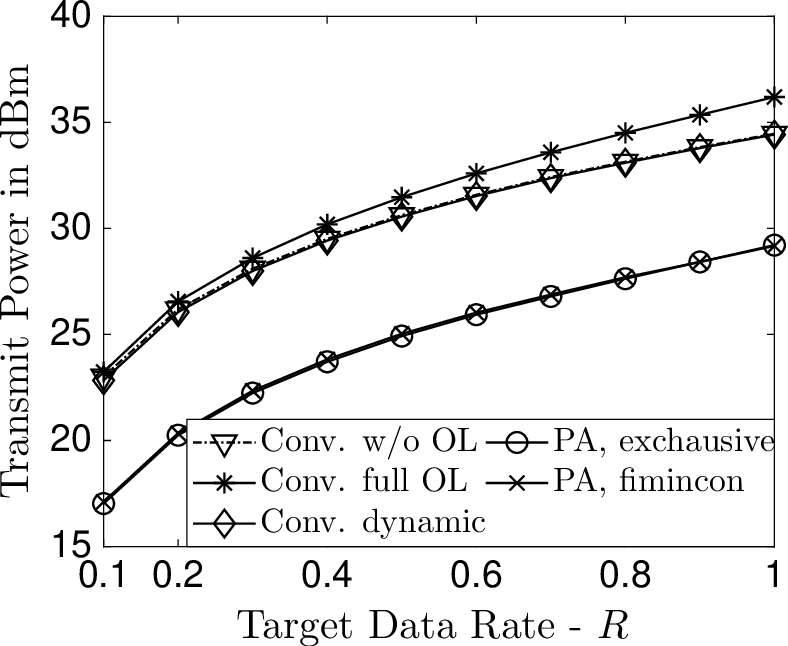}}   \vspace{-1em}
\end{center}
\caption{ Impact of pinching antennas on the transmit power for the considered traffic offloading schemes, where Case II is focused on and the scenario in Section \ref{section IV} is considered.   \vspace{-1em} }\label{figsim5}\vspace{-1em}
\end{figure}

   \begin{figure}[t] \vspace{-2em}
\begin{center}
\subfigure[Conventional antennas]{\label{figsim6a}\includegraphics[width=0.321\textwidth]{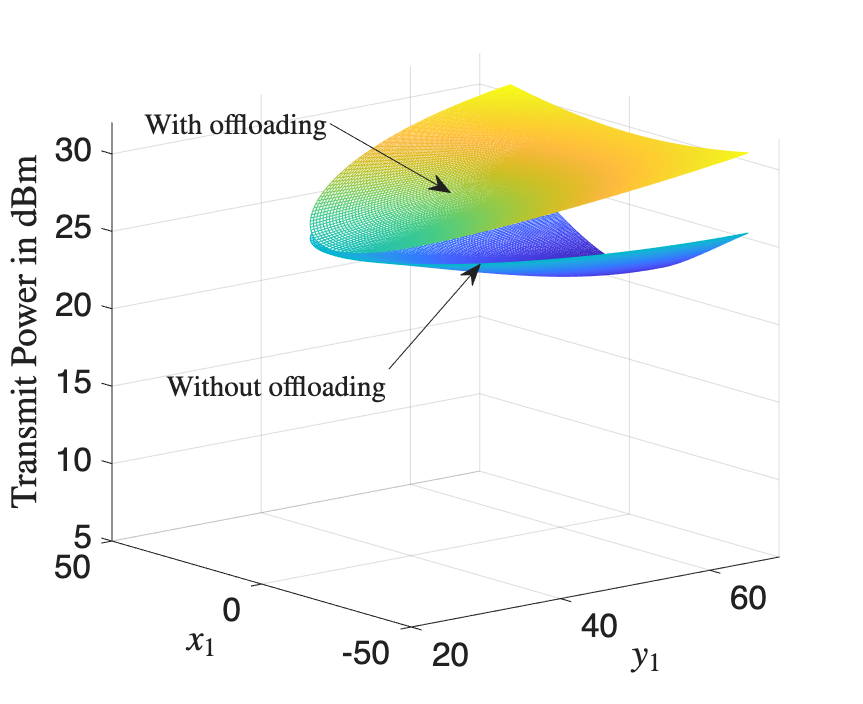}} 
\subfigure[Pinching antennas]{\label{figsim6b}\includegraphics[width=0.321\textwidth]{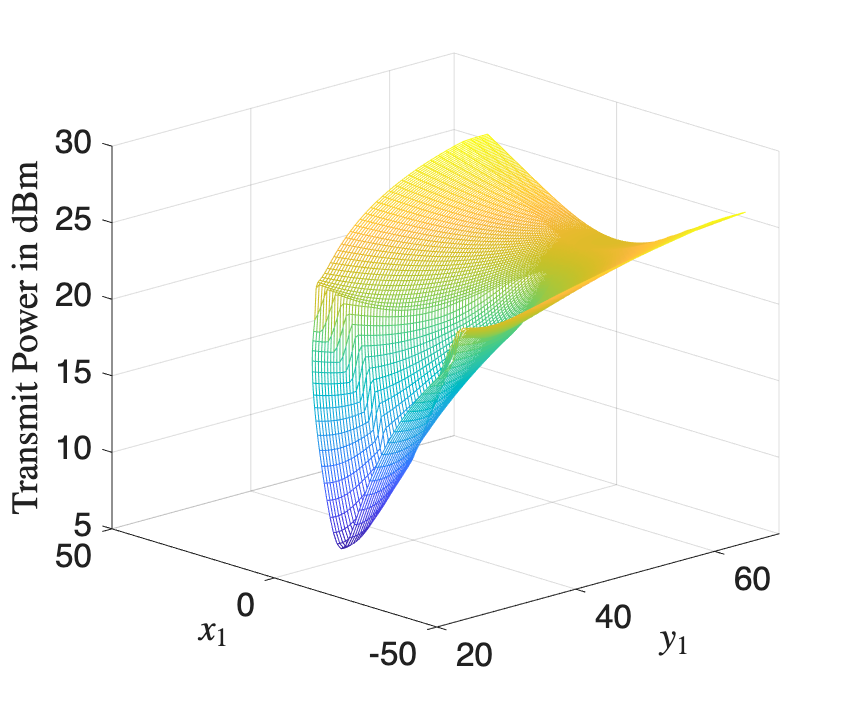}}   \vspace{-1em}
\end{center}
\caption{ Impact of pinching antennas on the transmit power for the considered traffic offloading schemes, where Case I is focused on and the scenario in Section \ref{section IV} is considered.  The same system parameter values as in Fig. \ref{figsim3} are used.   \vspace{-1em} }\label{figsim6}\vspace{-1em}
\end{figure}
\vspace{-1em}
\section{Conclusions}\label{section VI}
In this paper, the key features of pinching antennas, i.e., establishing strong connections between transceivers and making the physical boundaries of wireless cells reconfigurable, have been used to facilitate traffic offloading. An overall transmit power minimization problem has been formulated, where the optimal transmit powers and antenna locations have been obtained. Our analytical and simulation results demonstrate that the use of pinching antennas can efficiently support traffic offloading, achieve low energy consumption, and realize balanced cell resource utilization.

In this paper, it was assumed that only a single cell is equipped with a pinching-antenna system. The application of pinching antennas to general multi-cell scenarios, such as cloud radio access networks (C-RAN), fog radio access networks (F-RAN), and cell-free massive multiple-input multiple-output (MIMO) networks, is an important direction for future research \cite{CMCC,Dingpoor1311,7513863, 7827017}. 
\appendices
\section{Proof for Lemma \ref{lemma2}} \label{proof2}
 Problem \eqref{pb:13} is a convex optimization problem, and hence its optimal solution can be obtained based on the KKT conditions. The Lagrangian of problem \eqref{pb:13} is given by
\begin{align}
\mathcal{L} =&   P_{0m}+P_{mm}+P_{m} -\lambda_1P_{0m} - \lambda_2P_{mm}-\lambda_3P_m\\\nonumber&+
\lambda_4\left( \left(2^R-1\right)   g_{0m}P_{mm} - g_{0m}P_{0m}  +\left(2^R-1\right)\right)\\\nonumber
&+\lambda_5\left(
R-\log_2\left(
1+   g_{mm}  P_{mm}  
\right) -\log_2\left(
1+  g_m P_m 
\right)
\right),
\end{align}
 where $\lambda_m$, $1\leq m \leq 5$, denote the Lagrangian multipliers.

 The corresponding KKT conditions for problem \eqref{pb:13} are given by
 \begin{align}\label{kktcccc}
\left\{\begin{array}{l} 
 1  -\lambda_1   -
\lambda_4  g_{0m}  =0\\
1  - \lambda_2 +
\lambda_4  \left(2^R-1\right)   g_{0m}  - \frac{g_{mm}\lambda_5}{
\left(1+   g_{mm}  P_{mm}  \right)\log 2}  =0\\
1 -\lambda_3  -\frac{g_m\lambda_5 }{\left(1+  g_m P_m\right) \log 2}=0\\
\lambda_4\left( \left(2^R-1\right)   g_{0m}P_{mm} - g_{0m}P_{0m}  +\left(2^R-1\right)\right)=0\\
\lambda_5\left(
R-\log_2\left(
1+   g_{mm}  P_{mm}  
\right) -\log_2\left(
1+  g_m P_m 
\right)
\right)=0\\
\lambda_1P_{0m}=0,  \lambda_2P_{mm}=0,  \lambda_3P_m=0\\
\left(2^R-1\right)   g_{0m}P_{mm} - g_{0m}P_{0m}  +\left(2^R-1\right)
 \leq 0 \\  
R-\log_2\left(
1+   g_{mm}  P_{mm}  
\right) -\log_2\left(
1+  g_m P_m 
\right) \leq 0\\ P_{mm}\geq 0,P_m\geq 0, P_{0m}\geq 0
 \end{array}\right.,
\end{align}
where $\log(\cdot)$ denotes the natural logarithm. 
It is straightforward to show that $P_{0m}\neq 0$, which means that $\lambda_1=0$. The fact that $\lambda_1=0$ leads to the conclusions that  $\lambda_4=\frac{1}{g_{0m}}$ and $\left(2^R-1\right)   g_{0m}P_{mm} - g_{0m}P_{0m}  +\left(2^R-1\right)=0$. The discussions in Section \ref{section III} lead to the intuition that there could be three cases, which correspond to ${\rm U}_m$'s offloading strategies and hence are analyzed separately in the following.

\subsubsection{$P_m\neq0$ and $P_{mm}\neq0$}
In this case, ${\rm U}_m$ is simultaneouslly served by both ${\rm BS}_m$ and ${\rm BS}_0$. The KKT conditions shown in \eqref{kktcccc} can be simplified as follows:
 \begin{align}\label{kktcccc}
\left\{\begin{array}{l}  
 2^R   
     - \frac{g_{mm}\lambda_5}{
\left(1+   g_{mm}  P_{mm}\right)\log 2  }  =0,
1   -\frac{g_m\lambda_5 }{\left(1+  g_m P_m\right)\log 2}=0\\
  \left(2^R-1\right)   g_{0m}P_{mm} - g_{0m}P_{0m}  +\left(2^R-1\right) =0\\
\lambda_5\left(
R-\log_2\left(
1+   g_{mm}  P_{mm}  
\right) -\log_2\left(
1+  g_m P_m 
\right)
\right)=0\\
\lambda_1 =0,  \lambda_2 =0,  \lambda_3 =0,\lambda_4=\frac{1}{g_{0m}}  \\  
R-\log_2\left(
1+   g_{mm}  P_{mm}  
\right) -\log_2\left(
1+  g_m P_m 
\right) \leq 0\\ P_{mm}\geq 0,P_m\geq 0, P_{0m}\geq 0
 \end{array}\right..
\end{align}  
 It is straightforward to show that $\lambda_5\neq 0$, which means the three transmit power solutions can be obtained from the following linear equations:
  \begin{align}\label{kktcccc}
\left\{\begin{array}{l}  
  g_{mm}\left(1+  g_m P_m\right)    = g_m 2^R\left(
1+   g_{mm}  P_{mm} \right) \\
  \left(2^R-1\right)   g_{0m}P_{mm} - g_{0m}P_{0m}  +\left(2^R-1\right) =0\\
R-\log_2\left(
1+   g_{mm}  P_{mm}  
\right) -\log_2\left(
1+  g_m P_m 
\right)
=0 
 \end{array}\right.,
\end{align}  
  which leads to the following solutions:
\begin{align}\nonumber
P_{mm}=&\frac{1}{g_{mm}}\left(\sqrt{\frac{g_{mm}}{g_m}}-1\right),
\quad
P_m=\frac{1}{g_m}\left(2^R\sqrt{\frac{g_m}{g_{mm}}}-1\right),
 \\ \label{solutionx1}
P_{0m}=&\frac{2^R-1}{g_{0m}}\bigl(g_{0m}P_{mm}+1\bigr).
\end{align}
 The solutions in \eqref{solutionx1} yields $ 
\lambda_5=\frac{2^R\log 2 }{\,\sqrt{g_m g_{mm}}}\geq0
$, and therefore, the condition for the solutions in \eqref{solutionx1} to be optimal is $P_{mm}> 0$ and $P_m> 0$, i.e.,
\begin{align}
\sqrt{\frac{g_{mm}}{g_m}}-1> 0, \quad 2^R\sqrt{\frac{g_m}{g_{mm}}}-1> 0,
\end{align}
which means 
\begin{align}\label{condition1}
 1< \frac{g_{mm}}{g_m} < 2^{2R}.
\end{align}

\subsubsection{$P_m=0$ and $P_{mm}\neq0$}
This case corresponds to the scenario, where ${\rm U}_m$ is completely offloaded from ${\rm Cell}_m$ to ${\rm Cell}_0$. In this case, the KKT conditions can be simplified as follows:
 \begin{align}\label{kktcccc}
\left\{\begin{array}{l}  
2^R        - \frac{g_{mm}\lambda_5}{\left(
1+   g_{mm}  P_{mm} \right)\log 2 }  =0\\
1 -\lambda_3  -\frac{g_m\lambda_5}{\log 2} =0\\
  \left(2^R-1\right)   g_{0m}P_{mm} - g_{0m}P_{0m}  +\left(2^R-1\right) =0\\
\lambda_5\left(
R-\log_2\left(
1+   g_{mm}  P_{mm}  
\right)  
\right)=0\\
\lambda_1 =0,  \lambda_2 =0,   P_m=0,\lambda_4=\frac{1}{g_{0m}}\\ 
R-\log_2\left(
1+   g_{mm}  P_{mm}  
\right)   \leq 0\\ P_{mm}\geq 0,P_m= 0, P_{0m}\geq 0
 \end{array}\right..
\end{align}
It is straightforward to show that $\lambda_5\neq 0$, which leads to the following conclusions:
\begin{align}\label{kktcccc2}
\left\{\begin{array}{l}  
 P_{0m} =\left(2^R-1\right)\left[ \left(2^R-1\right)   \frac{1}{g_{mm}}   +\frac{1}{ g_{0m}}\right]  \\P_{mm} =\frac{2^{R}-1}{g_{mm}} ,P_{m}=0\\
\lambda_1 =0,  \lambda_2 =0,   P_m=0,\lambda_4=\frac{1}{g_{0m}} \\
\lambda_5=\frac{2^{2R}\log 2}{g_{mm}}  , 
\lambda_3=1   -g_m \frac{2^{2R}}{g_{mm} } 
 \end{array}\right..
\end{align}
Therefore, the condition for the optimality of this case is given by
\begin{align}\label{condition2}
\lambda_3\geq 0\Longrightarrow    \frac{g_{mm}}{g_m}\geq  2^{2R}. 
\end{align}

\subsubsection{$P_m\neq 0$ and $P_{mm}=0$}
This case corresponds to the scenario where ${\rm U}_m$ chooses to stay in ${\rm Cell}_m$. In this case, the KKT conditions can be simplified as follows:
 \begin{align}\label{kktcccc}
\left\{\begin{array}{l}  
2^R  - \lambda_2       -  \frac{g_{mm}\lambda_5}{\log 2}  =0\\
1   -\frac{g_m\lambda_5}{\left(1+  g_m P_m\right)\log 2}=0\\
  - g_{0m}P_{0m}  +\left(2^R-1\right) =0\\
\lambda_5\left(
R -\log_2\left(
1+  g_m P_m 
\right)
\right)=0\\
\lambda_1 =0,  \lambda_2P_{mm}=0,  \lambda_3 =0,\lambda_4=\frac{1}{g_{0m}}\\  
R  -\log_2\left(
1+  g_m P_m 
\right) \leq 0\\ P_{mm}= 0,P_m\geq 0, P_{0m}\geq 0
 \end{array}\right..
\end{align}
By using the fact that $\lambda_5\neq 0$, the optimal solution can be obtained as follows:
\begin{align}\label{kktcccc3}
\left\{\begin{array}{l}   
 P_{0m} =\frac{2^R-1}{g_{0m}} ,
P_m = \frac{2^R-1}{g_m} ,P_{mm}=0\\
\lambda_1 =0,   \lambda_3 =0,\lambda_4=\frac{1}{g_{0m}} 
\\   \lambda_2       =  2^R\left(1-\frac{ g_{mm} }{g_m } \right) ,
\lambda_5= \frac{2^R\log 2}{g_m }
 \end{array}\right..
\end{align}
Therefore, the condition for the optimality of this case is given by
\begin{align}\label{condition3}
 \lambda_2 \geq 0 \Longrightarrow  \frac{ g_{mm} }{g_m }  \leq 1. 
\end{align}
By combining \eqref{solutionx1}, \eqref{condition1}, \eqref{kktcccc2}, \eqref{condition2}, \eqref{kktcccc3}, \eqref{condition3}, the proof for Lemma \ref{lemma2} is complete.

\section{Proof for Lemma \ref{lemma3}}\label{proof3}
The lemma can be proved by studying the following two cases with different SIC decoding orders. 
\subsubsection{${\rm U}_{01}$ carries out SIC}
\label{subsubsection special 2}
For the considered special case, the constraint for complete offloading shown in \eqref{const 9 :9} is assumed to be always satisfied and hence can be omitted.  To facilitate the analysis of the optimal pinching antenna location, \eqref{const 9 :10} is first ignored, and hence problem \eqref{pb:9} can be simplified as follows:  
  \begin{problem}\label{pb:11} 
  \begin{alignat}{2}
 \underset{x^{\rm Pin}}{\rm{min}}  &\quad    2^R     \left( x_{0m}-x^{\rm Pin}\right)^2  +     \left( x_m-x^{\rm Pin}\right)^2  .
  \end{alignat}
\end{problem}  
Therefore, the optimal solution for the pinching antenna location is given by
\begin{align}\label{optimal1}
x^{\rm Pin*} = \frac{   2^R     x_{0m}+ x_m  }{ 2^R+1 },
\end{align}
which means that the minimal transmit power is given by
\begin{align}\nonumber
P^{\rm all} =&  \epsilon 2^R   \left[   \left(x_{01}- \frac{ \left(  2^R     x_{01}+ x_1\right) }{\left(2^R+1\right)}\right)^2+ y_{01} ^2+d^2\right] \\\nonumber &+   \epsilon \left[ \left(x_{1}- \frac{ \left(  2^R     x_{01}+ x_1\right) }{\left(2^R+1\right)}\right)^2+ y_{1} ^2+d^2\right]\\  =&\nonumber
 \epsilon \frac{2^R}{2^R+1}    \left(    x_{01}- x_1  \right)^2+ \epsilon2^R y_{01} ^2 \\ &+ \epsilon y_{1} ^2 +\epsilon(2^R+1)d^2 .\label{optimal12}
\end{align} 

In the following, we will show that the obtained optimal solution satisfies constraint \eqref{const 9 :10} (or equivalently \eqref{st8:2}), if $|y_{01}|\leq|y_1|$. In particular, with the optimal solution shown in \eqref{optimal1}, the difference between $\left| {\boldsymbol \psi} _{01} - {\boldsymbol \psi}^{\rm Pin}\right|^2$ and $ \left| {\boldsymbol \psi} _1 - {\boldsymbol \psi}^{\rm Pin}\right|^2$ shown in \eqref{st8:2} can be expressed as follows:
\begin{align}\label{sicconstraint1xd}
& \left| {\boldsymbol \psi} _{01} - {\boldsymbol \psi}^{\rm Pin}\right|^2- \left| {\boldsymbol \psi} _1 - {\boldsymbol \psi}^{\rm Pin}\right|^2\\\nonumber =&      \left[   \left(x_{01}- \frac{ \left(  2^R     x_{01}+ x_1\right) }{\left(2^R+1\right)}\right)^2+ y_{01}^2+d^2\right] \\\nonumber &-     \left[ \left(x_{1}- \frac{ \left(  2^R     x_{01}+ x_1\right) }{\left(2^R+1\right)}\right)^2+  y_{1}^2+d^2\right]
\\\nonumber=& 
  \left[   \left(  \frac{ \left(   x_{01}-x_1\right) }{\left(2^R+1\right)}\right)^2+  y_{01} ^2+d^2\right] \\\nonumber &-     \left[ 2^{2R}\left(  \frac{ \left(   x_1-  x_{01}\right) }{\left(2^R+1\right)}\right)^2+  y_{1}^2+d^2\right]
  \\\nonumber=& 
  -\frac{2^R-1}{2R+1} \left(   x_{01}-x_1\right)^2+   y_{01} ^2-  y_{1}^2    .
   \end{align}
   Therefore, the condition   $|y_{01}|\leq|y_1|$ guarantees  the constraint shown in \eqref{st8:2}, i.e., $\left| {\boldsymbol \psi} _{01} - {\boldsymbol \psi}^{\rm Pin}\right|^2\leq\left| {\boldsymbol \psi} _1 - {\boldsymbol \psi}^{\rm Pin}\right|^2$.

\subsubsection{${\rm U}_1$ carries out SIC}
Again, to facilitate the proof, the constraint shown in \eqref{const 16 :2} is ignored first. Since the constraint of \eqref{const 16 :1} is assumed to be satisfied for the considered special case, problem \eqref{pb:16} can be simplified as follows:
 \begin{problem}\label{pb:17} 
  \begin{alignat}{2}
 \underset{ x^{\rm Pin}}{\rm{min}}  &\quad       2^R   \left(
 x^{\rm Pin}-x_1
 \right)^2   +    \left(
 x^{\rm Pin}-x_{01}
 \right)^2  ,
  \end{alignat}
\end{problem}
whose optimal solution is given by
\begin{align}\label{optimal2}
 x^{\rm Pin} =  \frac{   2^R   x_1
    +    x_{01}
  }{ 2^R+1 }.
\end{align}

By using the above pinching antenna location, the minimal transmit power for this case is given by
 \begin{align}\label{optimal22}
 P^{\rm all}=& \epsilon 2^R   \left(x_{1}-  \frac{   2^R   x_1
    +    x_{01}
  }{ 2^R+1 }\right)^2+ \epsilon 2^Ry_{1} ^2+\epsilon 2^R d^2  \\\nonumber &+   \epsilon \left(x_{01}-  \frac{   2^R   x_1
    +    x_{01}
  }{ 2^R+1 }\right)^2+  \epsilon y_{01}^2+\epsilon d^2\\    =&  \epsilon
 \frac{   2^R}{ 2^R+1} \left(   x_{01}-x_1
   \right)^2 
+\epsilon2^R  y_1
^2 
  +\epsilon y_{01}
^2+\epsilon (2^R+1)d^2 .\nonumber
    \end{align}
   
In order to show that the obtained optimal solution satisifies the constraint shown in \eqref{const 16 :2}, we note that with the optimal pinching antenna location, the difference  between the distances, $ \left| {\boldsymbol \psi} _{01} - {\boldsymbol \psi}^{\rm Pin}\right|^2$ and $ \left| {\boldsymbol \psi} _1 - {\boldsymbol \psi}^{\rm Pin}\right|^2$ can be expressed as follows: 
\begin{align}
& \left| {\boldsymbol \psi} _{01} - {\boldsymbol \psi}^{\rm Pin}\right|^2- \left| {\boldsymbol \psi} _1 - {\boldsymbol \psi}^{\rm Pin}\right|^2\\\nonumber =&
 \left(
\frac{   2^R   x_1
    +    x_{01}
  }{ 2^R+1 }-x_{01}
 \right)^2 + y_{01}
^2     - \left(
\frac{   2^R   x_1
    +    x_{01}
  }{ 2^R+1 }-x_{1}
 \right)^2 -  y_{1} ^2   \\\nonumber =&  \frac{2^R-1}{2^R+1}
  \left( x_1
    -    x_{01}
 \right)^2  + y_{01} ^2    -  y_{1}
 ^2  .
\end{align}
Therefore, if $|y_{01}|> |y_1|$, the constraint shown in \eqref{const 16 :2} is satisfied.  
 
By combining \eqref{optimal1}, \eqref{optimal12}, \eqref{optimal2}, and \eqref{optimal22}, the concise expressions shown in Lemma \ref{lemma3} can be obtained and the proof is complete. 

\bibliographystyle{IEEEtran}
\bibliography{IEEEfull,trasfer}

\begin{thebibliography}{10}
\providecommand{\url}[1]{#1}
\csname url@samestyle\endcsname
\providecommand{\newblock}{\relax}
\providecommand{\bibinfo}[2]{#2}
\providecommand{\BIBentrySTDinterwordspacing}{\spaceskip=0pt\relax}
\providecommand{\BIBentryALTinterwordstretchfactor}{4}
\providecommand{\BIBentryALTinterwordspacing}{\spaceskip=\fontdimen2\font plus
\BIBentryALTinterwordstretchfactor\fontdimen3\font minus
  \fontdimen4\font\relax}
\providecommand{\BIBforeignlanguage}[2]{{%
\expandafter\ifx\csname l@#1\endcsname\relax
\typeout{** WARNING: IEEEtran.bst: No hyphenation pattern has been}%
\typeout{** loaded for the language `#1'. Using the pattern for}%
\typeout{** the default language instead.}%
\else
\language=\csname l@#1\endcsname
\fi
#2}}
\providecommand{\BIBdecl}{\relax}
\BIBdecl

\bibitem{you6g}
X.~You, C.~Wang, J.~Huang \emph{et~al.}, ``Towards {6G} wireless communication
  networks: {V}ision, enabling technologies, and new paradigm shifts,''
  \emph{Sci. China Inf. Sci.}, vol.~64, no. 110301, pp. 1--74, Feb. 2021.

\bibitem{irs1}
M.~D. Renzo, M.~Debbah, D.-T. Phan-Huy, A.~Zappone, M.-S. Alouini, C.~Yuen,
  V.~Sciancalepore, G.~C. Alexandropoulos, J.~Hoydis, H.~Gacanin, J.~de~Rosny,
  A.~Bounceu, G.~Lerosey, and M.~Fink, ``Smart radio environments empowered by
  {AI} reconfigurable meta-surfaces: An idea whose time has come,''
  \emph{EURASIP J. on Wirel. Com. Netw.}, vol. 129, pp. 1--20, May 2019.

\bibitem{irs2}
Q.~{Wu} and R.~{Zhang}, ``Intelligent reflecting surface enhanced wireless
  network via joint active and passive beamforming,'' \emph{IEEE Trans. Wirel.
  Commun.}, vol.~18, no.~11, pp. 5394--5409, Nov. 2019.

\bibitem{10318061}
L.~Zhu, W.~Ma, and R.~Zhang, ``Modeling and performance analysis for movable
  antenna enabled wireless communications,'' \emph{IEEE Trans. Wireless
  Commun.}, vol.~23, no.~6, pp. 6234--6250, Jun. 2024.

\bibitem{9264694}
K.-K. Wong, A.~Shojaeifard, K.-F. Tong, and Y.~Zhang, ``Fluid antenna
  systems,'' \emph{IEEE Trans. Wireless Commun.}, vol.~20, no.~3, pp.
  1950--1962, Mar. 2021.

\bibitem{pinching_antenna2}
A.~Fukuda, H.~Yamamoto, H.~Okazaki, Y.~Suzuki, and K.~Kawai, ``Pinching antenna
  - using a dielectric waveguide as an antenna,'' \emph{NTT DOCOMO Technical
  J.}, vol.~23, no.~3, pp. 5--12, Jan. 2022.

\bibitem{mypa}
Z.~Ding, R.~Schober, and H.~V. Poor, ``Flexible-antenna systems: A
  pinching-antenna perspective,'' \emph{IEEE Trans. Commun.}, vol.~73, no.~10,
  pp. 9236--9253, Oct. 2025.

\bibitem{11434944}
Y.~Xu, J.~Cui, Y.~Zhu, Z.~Ding, T.-H. Chang, R.~Schober, V.~W.~S. Wong, O.~A.
  Dobre, G.~K. Karagiannidis, H.~V. Poor, and X.~You, ``Generalized
  pinching-antenna systems: A tutorial on principles, design strategies, and
  future directions,'' \emph{IEEE Commun. Surveys Tuts.}, vol.~28, pp.
  5872--5908, 2026.

\bibitem{11202577}
Z.~Wang, C.~Ouyang, X.~Mu, Y.~Liu, and Z.~Ding, ``Modeling and beamforming
  optimization for pinching-antenna systems,'' \emph{IEEE Trans. Commun.},
  vol.~73, no.~12, pp. 13\,904--13\,919, Dec. 2025.

\bibitem{11348983}
C.~Ouyang, H.~Jiang, Z.~Wang, Y.~Liu, and Z.~Ding, ``Uplink and downlink
  communications in segmented waveguide-enabled pinching-antenna systems
  ({SWANs}),'' \emph{IEEE Trans. Commun.}, vol.~74, pp. 3688--3703, 2026.

\bibitem{11289518}
H.~Jiang, Z.~Wang, and Y.~Liu, ``Pinching-antenna system {(PASS)}-enhanced
  covert communications,'' \emph{IEEE J. Sel. Areas Commun.}, vol.~44, pp.
  1736--1750, 2026.

\bibitem{11215679}
K.~Wang, Z.~Ding, and N.~Al-Dhahir, ``Pinching-antenna systems for physical
  layer security,'' \emph{IEEE Wireless Commun. Lett.}, vol.~15, pp. 260--264,
  2026.

\bibitem{11303890}
E.~Illi, M.~Qaraqe, and A.~Ghrayeb, ``Secure pinching antenna-aided {ISAC},''
  \emph{IEEE Commun. Lett.}, vol.~30, pp. 727--731, 2026.

\bibitem{mynpj}
Z.~Ding, ``Pinching-antenna assisted {ISAC}: a {CRLB} perspective,'' \emph{npj
  Wirel. Technol.}, vol.~1, no.~4, 2025.

\bibitem{11212813}
W.~Mao, Y.~Lu, Y.~Xu, B.~Ai, O.~A. Dobre, and D.~Niyato, ``Multi-waveguide
  pinching antennas for {ISAC},'' \emph{IEEE Trans. Wireless Commun.}, vol.~25,
  pp. 5846--5858, 2026.

\bibitem{11314615}
M.~Liu, Y.~Xiao, J.~Chen, S.~Yang, X.~Lei, and M.~Xiao, ``Integrated sensing
  and communication with index modulation over pinching antennas,'' \emph{IEEE
  Commun. Lett.}, vol.~30, pp. 737--741, 2026.

\bibitem{11414134}
C.~Ouyang, Z.~Wang, Y.~Liu, and Z.~Ding, ``Rate region of {ISAC} for
  pinching-antenna systems,'' \emph{IEEE Trans. Commun.}, vol.~74, pp.
  5849--5866, 2026.

\bibitem{11360288}
B.~Wu, F.~Fang, M.~Zeng, and X.~Wang, ``Straggler-resilient federated learning
  over a hybrid conventional and pinching antenna network,'' \emph{IEEE Trans.
  Veh. Tech.}, pp. 1--6, 2026.

\bibitem{hinapa1}
S.~Asaad, H.~Tabassum, and P.~Wang, ``Pinching antennas-assisted low-latency
  federated learning over multi-user wireless networks,'' (submitted) Available
  on-line at arXiv:2603.08595.

\bibitem{yushen1}
Y.~Lin and Z.~Ding, ``Tail-latency-aware federated learning with pinching
  antenna: Latency, participation, and placement,'' \emph{IEEE Trans. Wireless
  Commun.}, (submitted) Available on-line at arXiv:2601.18097.

\bibitem{11195162}
Y.~Zhu, Z.~Ding, and X.~You, ``Topological perspective of large-scale
  multi-cell deployment of excitable waveguide dielectrics,'' \emph{IEEE
  Wireless Commun. Lett.}, vol.~15, pp. 151--155, 2026.

\bibitem{11315149}
Y.~Sun, Z.~Ding, and G.~K. Karagiannidis, ``A stochastic geometric analysis on
  multi-cell pinching-antenna systems under blockage effect,'' \emph{IEEE
  Wireless Commun. Lett.}, vol.~15, pp. 1085--1089, 2026.

\bibitem{mymulcell}
Z.~Ding, ``Toward a quiet wireless world: Multi-cell pinching-antenna
  transmission,'' \emph{IEEE Wireless Commun. Lett.}, (submitted) Available
  on-line at arXiv:2602.19459.

\bibitem{7012044}
X.~Chen, J.~Wu, Y.~Cai, H.~Zhang, and T.~Chen, ``Energy-efficiency oriented
  traffic offloading in wireless networks: A brief survey and a learning
  approach for heterogeneous cellular networks,'' \emph{IEEE J. Sel. Areas
  Commun.}, vol.~33, no.~4, pp. 627--640, 2015.

\bibitem{9061001}
N.~Kato, B.~Mao, F.~Tang, Y.~Kawamoto, and J.~Liu, ``Ten challenges in
  advancing machine learning technologies toward {6G},'' \emph{IEEE Wireless
  Commun.}, vol.~27, no.~3, pp. 96--103, 2020.

\bibitem{mojobabook}
M.~Vaezi, Z.~Ding, and H.~V. Poor, \emph{Multiple Access Techniques for {5G}
  Wireless Networks and Beyond}.\hskip 1em plus 0.5em minus 0.4em\relax
  Springer International Publishing, 2019.

\bibitem{7498101}
C.~Rosa, K.~Pedersen, H.~Wang, P.-H. Michaelsen, S.~Barbera, E.~Malkamäki,
  T.~Henttonen, and B.~Sébire, ``Dual connectivity for {LTE} small cell
  evolution: functionality and performance aspects,'' \emph{IEEE Commun. Mag.},
  vol.~54, no.~6, pp. 137--143, 2016.

\bibitem{7883844}
Y.~Wu and L.~P. Qian, ``Energy-efficient {NOMA}-enabled traffic offloading via
  dual-connectivity in small-cell networks,'' \emph{IEEE Commu. Lett.},
  vol.~21, no.~7, pp. 1605--1608, Jul. 2017.

\bibitem{6076617}
S.~Dimatteo, P.~Hui, B.~Han, and V.~O. Li, ``Cellular traffic offloading
  through wifi networks,'' in \emph{Proc. IEEE Eighth Int. Conf. on Mobile
  Ad-Hoc and Sensor Systems}, Valencia, Spain, 2011, pp. 192--201.

\bibitem{Cover1991}
T.~Cover and J.~Thomas, \emph{Elements of Information Theory}, 6th~ed.\hskip
  1em plus 0.5em minus 0.4em\relax Wiley and Sons, New York, 1991.

\bibitem{closedformzid}
Z.~{Ding} and H.~V. Poor, ``Analytical optimization for antenna placement in
  pinching-antenna systems,'' \emph{IEEE Trans. Wireless Commun.}, (submitted)
  Available on-line at arXiv:2507.13307, 2025.

\bibitem{Boyd}
S.~Boyd and L.~Vandenberghe, \emph{Convex Optimization}.\hskip 1em plus 0.5em
  minus 0.4em\relax Cambridge University Press, Cambridge, UK, 2003.

\bibitem{9679390}
Z.~Ding, D.~Xu, R.~Schober, and H.~V. Poor, ``Hybrid {NOMA} offloading in
  multi-user {MEC} networks,'' \emph{IEEE Trans. Wireless Commun.}, vol.~21,
  no.~7, pp. 5377--5391, 2022.

\bibitem{CMCC}
``{C-RAN}: The road towards green {RAN},'' China Mobile Res. Inst., Beijing,
  China, Oct. 2011, White Paper, ver. 2.5.

\bibitem{Dingpoor1311}
Z.~Ding and H.~V. Poor, ``The use of spatially random base stations in cloud
  radio access networks,'' \emph{IEEE Signal Process. Lett.}, vol.~20, no.~11,
  pp. 1138--1141, Nov 2013.

\bibitem{7513863}
M.~Peng, S.~Yan, K.~Zhang, and C.~Wang, ``Fog-computing-based radio access
  networks: issues and challenges,'' \emph{IEEE Network}, vol.~30, no.~4, pp.
  46--53, 2016.

\bibitem{7827017}
H.~Q. Ngo, A.~Ashikhmin, H.~Yang, E.~G. Larsson, and T.~L. Marzetta,
  ``Cell-free massive {MIMO} versus small cells,'' \emph{IEEE Transactions on
  Wireless Communications}, vol.~16, no.~3, pp. 1834--1850, 2017.

\end{thebibliography}
  \end{document}